\documentclass[accepted]{uai2023} 

\usepackage[american]{babel}

\usepackage{natbib} 
\usepackage{mathtools} 
\usepackage{booktabs} 
\usepackage{tikz} 

\usepackage{amssymb}
\usepackage{bm}
\usepackage{tikz}
\usetikzlibrary{quantikz}
\usepackage[ruled,linesnumbered]{algorithm2e}
\usepackage{algorithmic}
\usepackage[switch]{lineno}
\usepackage{accents}
\usepackage{amsthm}

\newtheorem{theorem}{Theorem}
\newtheorem{lemma}{Lemma}

\newtheorem{corollary}{Corollary}
\newtheorem{proposition}{Proposition}

\newcommand{\BibTeX}{\rm B\kern-.05em{\sc i\kern-.025em b}\kern-.08em\TeX}
\newcommand{\zero}{|0\rangle}
\newcommand{\one}{|1\rangle}
\newcommand{\x}{|x\rangle}
\newcommand{\hn}{\accentset{\circ}{n}}
\newcommand{\hist}{\text{\rm \textbf{hist}}}
\newcommand{\hhist}{\widehat{\hist}}
\newcommand{\thist}{\overline{\hist}}

\newcommand{\MOV}{\text{\rm MoV}}

\newcommand{\llceil}{\left\lceil}
\newcommand{\rrceil}{\right\rceil}
\newcommand{\lnorm}{\left|\left|}
\newcommand{\rnorm}{\right|\right|}
\newcommand{\qps}{\rm QPS}
\newcommand{\tone}{s}

\newcommand{\vv}{\boldsymbol{V}}

\newcommand{\prc}{\Pr\left[\text{\rm correct}\right]}
\newcommand{\nwi}{n_{\text{\rm win}}}
\newcommand{\nlo}{n_{\text{\rm lose}}}
\newcommand{\cp}{{\text{\rm Cap}}}

\title{Accelerating Voting by Quantum Computation}

\author[1]{\href{mailto:<aoliu.cs@gmail.com>?Subject=Your UAI 2023 paper}{Ao Liu}{}}
\author[1]{Qishen Han}
\author[1]{Lirong Xia}
\author[2]{Nengkun Yu}
\affil[1]{%
    Department of Computer Science\\
    Rensselaer Polytechnic Institute\\
    Troy, NY, USA
}
\affil[2]{%
    Department of Computer Science\\
    Stony Brook University\\
    Stony Brook, NY, USA
}
  
\begin{document}

\maketitle

\begin{abstract}
Studying the computational complexity and designing fast algorithms for determining winners under voting rules are classical and fundamental questions in computational social choice. In this paper, we accelerate voting by leveraging quantum computation: we propose a quantum-accelerated voting algorithm that can be applied to any anonymous voting rule. We  show that our algorithm can be quadratically faster than any classical algorithm (based on sampling with replacement) under a wide range of common voting rules, including positional scoring rules, Copeland, and single transferable voting (STV). Precisely, our quantum-accelerated voting algorithm outputs the correct winner with high probability in  $\Theta\left(\frac{n}{\MOV}\right)$ time, where $n$ is the number of votes and $\MOV$ is {\em margin of victory}, the smallest number of voters to change the winner. In contrast, any classical voting algorithm based on sampling with replacement requires  $\Omega\left(\frac{n^2}{\MOV^2}\right)$ time under a large class of voting rules. Our theoretical results are supported by experiments under plurality, Borda, Copeland, and STV. 

\end{abstract}


\section{Introduction}
Driven by the critical public need  of revolutionalizing modern democratic systems~\citep{mancini2015time,brill2018interactive}, voting has been   widely applied in many collective decision-making scenarios beyond political elections. Examples include search engines~\citep{dwork2001rank}, crowdsourcing~\citep{mao2013better}, database management~\citep{Belardinelli_2019}, and blockchain governance~\citep{grossi2022social}, just to name a few. In such large-scale, high-frequency collective decision-making scenarios, it is desirable that the winner is computed as soon as possible, ideally in sub-linear time in the number of votes, and perhaps at a (small) cost of its correctness. In fact, the study of computational complexity and algorithmic aspects of voting rules has been a key topic of {\em computational social choice}~\citep{brandt2016handbook}.

One natural approach is to randomly sample a subset of votes (with or without replacement) and compute the winner of the sampled votes. The idea can be dated back to Venetian elections in the 13th century~\citep{walsh2012lot} and has recently attracted much attention from the computational social choice community~\citep{Bhattacharyya2021sample,Flanigan2020sortition,flanigan2021fair}. However, the performance of a sampling algorithm is restricted by the number of samples needed to guarantee a certain level of correctness, which determines its runtime. Specifically, the runtime of a sampling algorithm would be quadratically related to the number of votes (see Table~\ref{tab:summary1}). Is there a faster algorithm, for example, sub-linear to the number of votes, that still preserves a high probability of correctness? 

Quantum computation appears to be a promising approach, as it has successfully accelerated many computational tasks such as  search~\citep{10.1145/237814.237866}, optimization~\citep{hogg2000quantum}, and machine learning~\citep{PhysRevA.94.022308,AJAGEKAR2020107119,AJAGEKAR2021117628}. However, we are not aware of previous work on accelerating voting using quantum computation. Thus, the following problem remains open. 
\begin{center}
    \textbf{Can voting be accelerated by quantum computation?}
\end{center}

We address this question with YES both theoretically and experimentally. We accelerate voting by designing a sub-linear quantum-accelerated voting algorithm where a small probability of ``errors'' is allowed, which outperforms the classical sampling-based voting algorithms.
Our contributions are three-fold. First, we propose the quantum-accelerated voting algorithm (Algorithm~\ref{alg:quantum_multi}). Our algorithm leverages the widely-used techniques of quantum counting~\citep{Brassard1998counting} to generate the histogram of the votes. The simple architecture guarantees that our algorithm will be easy to implement in the future. Second, we theoretically prove that our algorithm is quadratically faster than any classical sampling-based algorithm for many common voting rules, including positional scoring rule, STV, Copeland, and maximin (see Table~\ref{tab:summary1}). Third, we experimentally verify our theoretical results under the plurality, Borda, Copeland, and STV (Section~\ref{sec:exp}). In Section~\ref{sec:heuristic}, we provide heuristics that may further improve the performance of quantum-accelerated voting. 

\begin{table}[htp]
    \centering
    \begin{tabular}{cccc}
    \toprule
         & Runtime & Space Requirement\\ \hline
        {Quantum} (Thm.~\ref{thm:gsr}) &  {$\Theta\left(\frac{n}{\MOV}\right)$} & {$\Theta\left(\log(\frac{n}{\MOV})\right)$}\\
        Classical (Thm.~\ref{thm:classical}) & $\Omega\left(\frac{n^2}{\MOV^2}\right)$ & $\Omega\left(\log(\frac{n^2}{\MOV^2})\right)$\\
    \bottomrule
    \end{tabular}
    \caption{Summary for the theoretical results, where $n$ is the number of votes, and $\MOV$ is the margin of victory. All results in the table assume a constant error rate (\emph{e.g.,} $1\%$).}\label{tab:summary1}
\end{table}

Our quantum-accelerated voting algorithm accelerates voting most significantly in the case where the margin of victory is sub-linear in $n$, \emph{e.g.,}  $\MOV = \Theta(n^c)$ with $c \in (0, 1)$. In this case, $\MOV$ is relatively small compared with $n$, and a $\Theta (\frac{n}{\MOV})$ acceleration from classical to quantum algorithm is significant. The ratio $\frac{n}{\MOV} = \Theta ( n^{1 -c})$ implies that the algorithm is sub-linear. See Section~\ref{sec:compare} for detailed discussions. 

 
\noindent\textbf{Related works and discussions. } To the best of our knowledge, our work is the first to use quantum computation to accelerate voting. \citet{vaccaro2007quantum}introduced the idea of quantum  computation to voting. The quantum voting algorithm in \citet{vaccaro2007quantum} provides security guarantees (against colluding attacks~\citep{lian2009handbook}). \citet{xue2017simple} improved the result in~\citet{vaccaro2007quantum} by proposing a simpler voting protocol but with stronger security guarantees. \citet{khabiboulline2021efficient} focuses on achieving anonymity without losing security guarantees. However, all approaches above require $\Omega(n)$ quantum communication cost and thus take $\Omega(n)$ time.

There is a large literature on efficient (classical) algorithms for the winner determination problem. \citet{Wang2019PUT} purposed fast algorithms to compute winners in ranked pairs and STV under parallel-universes tiebreaking~\citep{conitzer2009preference}, which is known to be NP-complete. Various papers have shown that the winner of Dodgson rule, while is NP-hard to compute in the worst case~\citep{bartholdi1989voting}, can be efficiently computed with high probability when the ranking is generated $i.i.d.$~\citep{mccabe2008approximability, homan2009guarantees} and under semi-random models~\citep{xia2022beyond}. This line of work focuses on designing algorithms for NP-hard winner determination problems, while our paper focuses on further accelerating voting whose winner determination problem is in P. 

\section{Preliminaries}\label{sec:prelim}
\textbf{VOTING.}\\
In voting, $n > 1$ voters cast their votes on $m > 1$ candidates. The candidates are denoted as $c_1,\cdots,c_m$. A vote represents a voter's preference towards the candidates, which is a full-ranking (linear order) over candidates. Since there are $m!$ types of full rankings for $m$ candidates, a vote can be represented as an $m!$-dimensional unit vector. That means if the vote is the $j$-th type, the $j$-th dimension of the vector is $1$, and all other dimensions are $0$'s. The vote of the $i$-th voter is denoted as $\vv_i = \left(V_{i,1},\cdots,V_{i,m!}\right)$. For example, if the $i$-th vote is the $j$-th type, we have that $\vv_{i,j} = 1$ and $V_{i,j'} = 0$ for all $j'\neq j$. A profile $P$ is a collection of $n$ voters' rankings. A voting rule $r$ is a mapping from the profile $P$ to the winner(s) among $m$ candidates. 

A histogram $\hist$ is an $m!$ dimension vector that records the number of each type of ranking in the profile. We use $r(\hist)$ to denote the winner under an anonymous voting rule $r$ and a profile with histogram $\hist$. In this paper, we assume that the voting rule $r$ satisfies {\em anonymity} and {\em canceling out}~\citep{ao2020private}. An anonymous voting rule selects the winner only based on the histogram of the profile and does not depend on the identity of the voter. A voting rule satisfies canceling out if the winner does not change after adding one copy of each ranking to the profile. Most common voting rules satisfy both anonymity and canceling-out, such as plurality, Borda, Copeland, and STV.

\noindent\textbf{Plurality.} The candidate ranked top in the most number of votes is chosen as the winner. 

\noindent\textbf{Borda.} Firstly, the rule calculates the Borda score of each candidate. A candidate ranked $i$-th in a vote gains a score of $(m-i)$ from that vote. For example, the Borda scores for the vote $[c_1 \succ c_2 \succ c_3 \succ c_4]$ are $\{c_1:3,\, c_2:2,\, c_3:1,\, c_4:0\}$. The Borda score of a candidate is the sum of scores it gains from each vote. Then, the candidate with the largest Borda score is chosen as the winner. 

\textbf{Copeland. } Copeland compares each pair of candidates, where the winner gets $1$ point, and the loser gets $0$ point. If two candidates are tied, both get $0.5$ points. After finishing all pairwise comparisons, the candidate receiving the most points is chosen as the winner.

\textbf{Single transferable vote (STV). } STV is an $(m-1)$-round voting rule. In each round, the candidate receiving the least number of top-ranked votes is eliminated. When all $(m-1)$ rounds are finished, the remaining candidate is chosen as the winner.

\textbf{Margin of victory.} \emph{Margin of victory} (MoV) is the smallest number $k$ such that there exist a set of $k$ voters who can change the winner by voting differently.\\

\textbf{QUANTUM COMPUTATION. }\\
\noindent\textbf{Quantum counting algorithm\footnote{This paper adopts the same notation system as~\citet{nielsen2002quantum}, which is a textbook about quantum computation.}~\citep{Brassard1998counting}}. Quantum counting is one of the key parts of our proposed quantum-accelerated voting algorithm. It counts the number of solutions to a search problem. Given a binary function $f: \{0,1, \cdots, 2^t -1\} \to \{0, 1\}$, the quantum counting circuit for $f$ computes the number of $x \in \{0,1, \cdots, 2^t -1\}$ such that $f(x) = 1$. A quantum counting circuit uses $t$ quantum bits to encode the function and $\tone$ quantum bits to calculate and record the output. More specifically, supposing $n_1$ is the number of $x$ such that $f(x) = 1$. The quantum counting circuit with $t + \tone$ quantum bits outputs a binary decimal $\hat{\varphi} = 0.b_1b_2\cdots b_s$, which estimates $\varphi = \arcsin\left(\sqrt{n_1\cdot 2^{-t}}\right)/\pi$. For example, binary decimal $0.011$ represents $(2^{-2} + 2^{-3}) = 3/8$. 
In the next lemma, we present three useful properties of quantum counting. 

\begin{lemma}[Properties of quantum counting]\label{lem:q_count_biased}
The quantum counting algorithm has the following three properties\\
1 (Tail bound, Inequality (5.34) in \citep{nielsen2002quantum}). For any $\delta > 2^{-\tone}$,
\begin{equation}\nonumber
\Pr[|\hat{\varphi} - \varphi| \geq 2^{-s} + \delta] \leq \frac{1}{2(\delta\cdot 2^{\tone}-1)}.
\end{equation}
2. Its runtime is $\Theta(2^\tone)$.\\
3. Its space requirement is $(t+\tone)$ quantum bits (qubits).
\end{lemma}

The first property in Lemma~\ref{lem:q_count_biased} says that a larger $\tone$ provides a stronger theoretical guarantee on the accuracy of quantum counting. However, a larger $\tone$ also corresponds to longer runtime. The detailed implementation of quantum counting and the reasoning behind why it accelerates can be found in Appendix~\ref{apx:quantum}. 

\noindent\textbf{Applying quantum counting to voting.}
We apply the quantum counting algorithm to estimate the histogram of a profile. For any $j\in\{1,\cdots,m!\}$, we estimate the $j$-th type of votes by setting the following binary function:
\begin{equation}\label{equ:fj}
f_j(x) = \left\{
\begin{array}{ll}
1 &  
\begin{array}{ll}\text{if candidate } x\text{'s vote}\\
\text{is of the } j \text{-th type}
\end{array}\\
& \\
0 & \text{otherwise} 
\end{array}\right..
\end{equation}
The number of $x$'s with  $f_j(x) = 1$ is  the number of votes of $j$-th type, \emph{i.e.},  $\hist_j$. The histogram of the votes is generated by enumerating all the $j$'s. We assume that the information in $\hist_j$ is stored in quantum RAM, which means the information can be efficiently encoded into quantum circuits ~\citep{giovannetti2008arch,giovannetti2008quantum,park2019circuit}.

\section{Quantum-Accelerated Voting Algorithm}
\label{sec:proto}

\noindent\textbf{Formal definition of quantum-accelerated voting. } We formally define quantum-accelerated voting in Algorithm~\ref{alg:quantum_multi}. In classical voting, votes are usually sent to an ``aggregator'', who is responsible for aggregating the votes and announcing the winner. Our quantum-accelerated voting follows a similar procedure, where the ``aggregator'' uses an algorithm accelerated by quantum computation. Basically, Algorithm~\ref{alg:quantum_multi} repeats the quantum counting algorithm by $K$ rounds. In each round, quantum counting estimates the histogram of the profile $\hist$ and applies the voting rule to the estimated histogram $\hhist$ to compute an estimated winner $c^{(k)} = r(\hhist)$. Then, the classical plurality voting rule is used to aggregate the estimated winners of $K$ rounds, \emph{i.e.} $c^{(1)},\cdots,c^{(K)}$. That is: the ``aggregator'' announces the candidate who is the estimated winner of the most rounds as the (final) winner. A tie-breaking rule (\emph{e.g.,} random tie-breaking) is applied when there are multiple candidates with the most rounds as the estimated winner. In Section~\ref{sec:theo_quantum}, we will show that Algorthm~\ref{alg:quantum_multi} (a combination of quantum and classical algorithms) has a better runtime than the quantum algorithm only. 

\begin{algorithm}[htp]
\caption{Quantum-Accelerated Voting Algorithm}\label{alg:quantum_multi}
\begin{algorithmic}[1]
\STATE {\bfseries Inputs:} $n$ voters' votes $\vv_0, \cdots, \vv_{n-1}$, a voting rule $r$, number of iteration $K$, and the number of qubits $\tone \geq 2$
\STATE {\bfseries Initialization:} Construct the binary functions $f_1, f_2, \cdots, f_{m!}$ based on $\vv_0, \cdots, \vv_{n-1}$
\FOR{$k\in\{1,\cdots,K\}$}
\STATE Initialize an $m!$-dimensional vector $\hhist$
\FOR{$j\in\{1,\cdots,m!\}$}
\STATE Construct and apply quantum counting circuit for $f_j$ with $\tone$ qubits. Denote the output of the quantum counting circuit as $0.b_1\cdots b_{\tone}$.
\STATE {Set the $j$-th component of $\hhist$ as $2^t\cdot \sin^2\left(\pi\cdot 0.b_1\cdots b_{\tone}\right)$}
\ENDFOR
\STATE Set $c^{(k)} = r(\hhist)$ as the winner of the $k$-th round
\ENDFOR 
\STATE {\bfseries Output} the (classical) plurality winner of $c^{(1)},\cdots,c^{(K)}$. 
\end{algorithmic}
\end{algorithm}

\noindent\textbf{Construct binary functions.} 
We construct a binary function $f_j$ to store the information about the $j$-th type of vote. $f_j(x)=1$ if and only if the $x$-th voter casts the $j$-th type of vote. Formally, $f_j:\{0,1,\cdots, 2^t -1\} \to \{0, 1\}$ is defined according to Equation (\ref{equ:fj}),  
where $t = \lceil \log n \rceil$ is the number of bits to encode $n$, and the voters are numbered from $0$ to $(n-1)$. For $x \in \left\{ n, \cdots,  2^{t}-1\right\}$, we just set $f_j(x) = 0$ for padding. Then, the number of $x$ such that $f_j(x) = 1$ is exactly the number of votes of the $j$-th type, \emph{i.e.} $\hist_j$. 

\noindent\textbf{Count the histogram.}
For each type $j$, a quantum counting circuit is constructed and used to estimate $\hist_j$. The output of the counting circuit is $\hat{\varphi} = 0.b_1b_2\cdots b_{\tone}$, which estimates $\varphi = \arcsin\left(\sqrt{\hist_j\cdot 2^{-t}}\right)/\pi$. Therefore, $\hhist_j = 2^t\cdot \sin^2\left(\pi\cdot 0.b_1\cdots b_{\tone}\right)$ estimates $\hist_j$. By enumerating all types $j$, we get $\hhist$, which estimates the histogram $\hist$.

\noindent\textbf{Decide the winner.} The quantum counting procedure (and the voting rule) runs for $K$ rounds. Each round's estimated winner is computed according to the estimated histogram $\hhist$. Finally, the quantum-accelerated voting algorithm outputs the candidate that wins in the most number of rounds.

\section{Theoretical Analysis of Quantum-Accelerated Voting}
\label{sec:theo_quantum}
In this section, we  provide theoretical guarantees about Algorithm~\ref{alg:quantum_multi}'s $\prc$ \footnote{Throughout this paper, we use $\prc$ to represent the probability (for an algorithm) to output the correct winner. Formally, $\prc \triangleq \Pr[\text{the algorithm's output} = r(P)]$, where $P$ is the voting profile and $r$ is the voting rule.}, runtime, and space requirements. 
In our analysis, the number of candidates $m$ is fixed.  Figure~\ref{fig:thm_chain} illustrates a roadmap of the results in this section. 

\begin{figure}[ht]
    \centering
    \includegraphics[width = 0.48\textwidth]{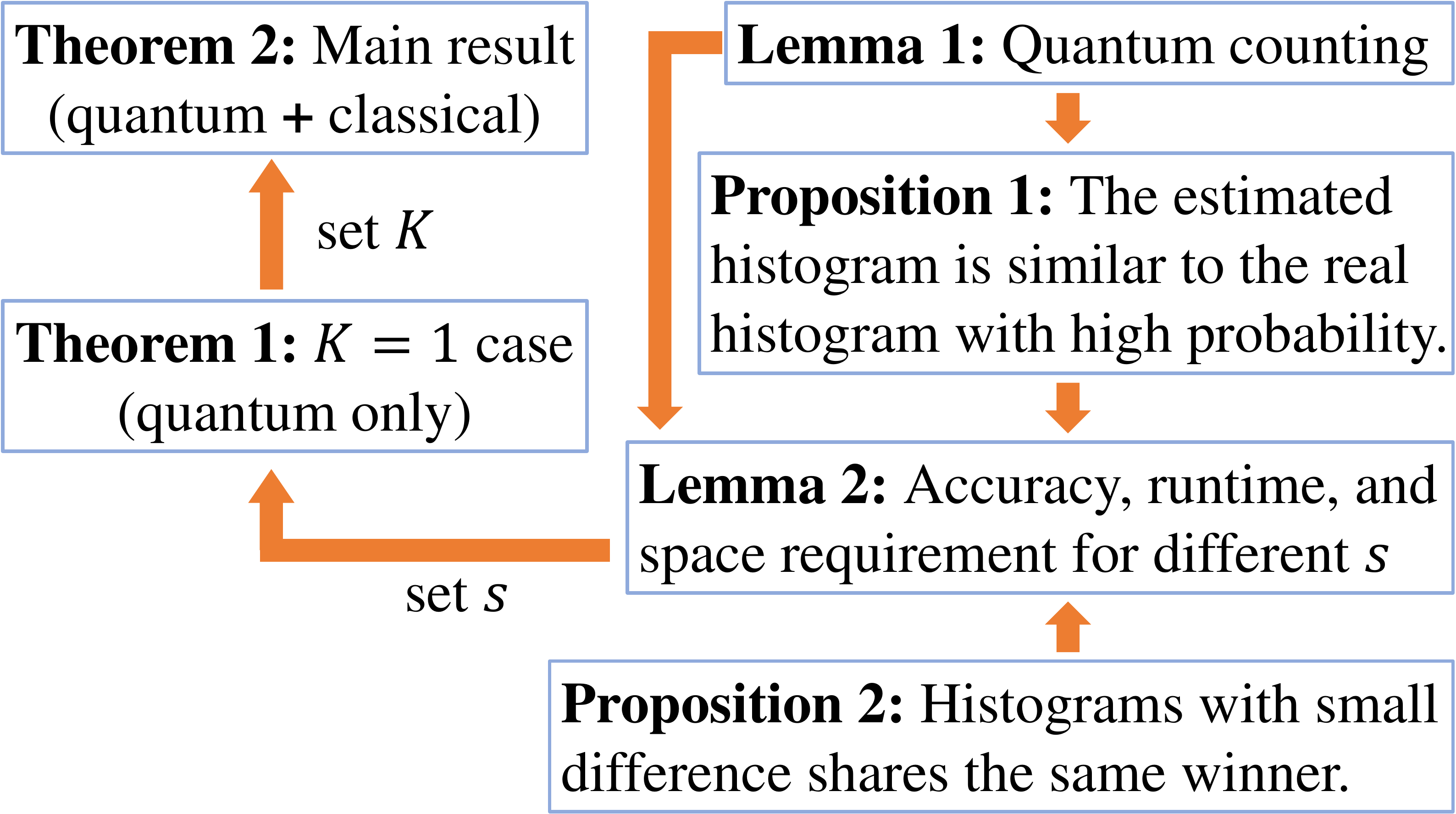}
    \caption{The logical chain of the theorems, lemmas, and propositions about Algorithm~\ref{alg:quantum_multi}. For example, the arrow from Theorem~\ref{lem:gsr_K=1} to Theorem~\ref{thm:gsr} represents that Theorem~\ref{thm:gsr} is proved based on Theorem~\ref{lem:gsr_K=1}.}
    \label{fig:thm_chain}
\end{figure}

In Theorem~\ref{lem:gsr_K=1}, we show the theoretical guarantee of Algorithm~\ref{alg:quantum_multi} when $K=1$. Note that when $K = 1$, the outer for loop in Algorithm~\ref{alg:quantum_multi} only runs for one round, which means the classical plurality voting rule directly outputs $c^{(1)}$. In other words, Algorithm~\ref{alg:quantum_multi} only contains quantum counting (and the voting rule $r$) when $K=1$. To simplify notations, we define a function $\sigma(\cdot)$ as follows.
{\small
\begin{equation}\nonumber
\begin{split}
\sigma(\varepsilon) &= 2+ \llceil t+ \log\Big(\frac{m!}{2\varepsilon}+2\Big) +\log \big(\pi\cdot(m!)\big) - \log(\MOV) \rrceil\\
&= \Theta\left(\log\big(\frac{n}{\varepsilon\cdot\MOV}\big)\right). 
\end{split}
\end{equation}
}
Throughout this paper, the logarithm function $\log(\cdot)$ represents $\log_2(\cdot)$ and $\ln(\cdot)$ represents $\log_e(\cdot)$.

\begin{theorem}[Theoretical guarantee of Algorithm~\ref{alg:quantum_multi} when $K=1$]
\label{lem:gsr_K=1}
For any given $\varepsilon\in (0,1) $, Algorithm~\ref{alg:quantum_multi} with $K=1$ and $\tone = \sigma(\varepsilon)$ has the following three properties. \\
1. Its $\prc \geq 1-\varepsilon$.\\
2. Its runtime is $\Theta\left(\frac{n}{\varepsilon \cdot \MOV}\right)$.\\
3. Its space requirement is  $\Theta\left(\log(\frac{n}{\varepsilon\cdot\MOV})\right)$.
\end{theorem}
For readability, we will present the proof of Theorem~\ref{lem:gsr_K=1} after the proof of Theorem~\ref{thm:gsr}. In Theorem~\ref{thm:gsr}, we will present the theoretical guarantee of Algorithm~\ref{alg:quantum_multi} when parameter $K$ is set properly. Under this setting, Algorithm~\ref{alg:quantum_multi} contains both quantum and classical aggregation methods. Comparing Theorem~\ref{lem:gsr_K=1} with Theorem~\ref{thm:gsr}, we know the (classical) plurality voting can reduce the $(1/\epsilon)$ term in runtime to $\log(1/\epsilon)$.

\begin{theorem}[Theoretical guarantee of Algorithm~\ref{alg:quantum_multi}]\label{thm:gsr}
For any given $\varepsilon\in (0,1) $, Algorithm~\ref{alg:quantum_multi} with $K=\lfloor24\ln(1/\varepsilon)\rfloor+1$ and $\tone = \sigma(1/4) = \Theta\big(\log(n/\MOV)\big)$ has the following three properties.\\
1. Its $\prc \geq 1-\varepsilon$.\\
2. Its runtime is $\Theta\left(\frac{n\log(1/\varepsilon)}{\MOV}\right)$.\\
3. Its space requirement is  $\Theta\left(\log(\frac{n\log(1/\varepsilon)}{\MOV})\right)$.
\end{theorem}
\begin{proof}
Setting $\tone = \sigma(1/4)$, we know that Algorithm~\ref{alg:quantum_multi} outputs the correct winner in each round with at least $p = \frac34$ probability according to Theorem~\ref{lem:gsr_K=1}. By Chernoff bound,
\begin{equation}\nonumber
\begin{split}
    \prc & \geq \Pr[r(P)\text{ wins in more than  } K/2 \text{ rounds}]\\
    & > 1 - \exp\left(-\Big(1-\frac{1}{2p}\Big)^2\cdot K \cdot \frac{p}{2}\right).
\end{split}
\end{equation}
Thus, $\prc \geq 1-\varepsilon$ holds when
\begin{equation}\nonumber
K > \frac{2p\cdot\ln(1/\varepsilon)}{(p-1/2)^2} = 24\ln(1/\varepsilon).
\end{equation}
Then, Theorem~\ref{thm:gsr} follows by the monotonicity of Chernoff bound towards $K$.
\end{proof}

Now, we have presented the two main theorems about Algorithm~\ref{alg:quantum_multi} and are ready to show the full proof of Theorem~\ref{lem:gsr_K=1}.

\begin{proof}[Proof~of~Theorem~\ref{lem:gsr_K=1}]
    Theorem~\ref{lem:gsr_K=1} follows by setting $s=\sigma(\varepsilon)$ for the following lemma, which bounds $\prc$, runtime, and space complexity of quantum-accelerated voting for different settings on the numbers of qubits $\tone$.
    \begin{lemma}\label{thm:hist_same} 
        For any $\tone \ge 2$, quantum-accelerated voting (Algorithm~\ref{alg:quantum_multi}) with $K=1$ has the following three properties.\\
        1. Its $\prc\geq1 - \frac{m!}{2(\delta\cdot 2^{\tone} - 1)}$ for any $ \delta \in (2^{-\tone}, \frac{\MOV}{2^{t+1}\pi\cdot(m!)} -2^{-\tone})$\\
        2. Its runtime is $\Theta(2^{\tone})$\\
        3. Its space requirement is $\Theta(\log(n) + \tone)$. \\
        Note that Property 1 holds only when the feasible region of $\delta$ is non-empty, which sets a constraint for $\tone$.
    \end{lemma}

\begin{proof}
\noindent\textbf{Runtime and space requirement.} The runtime and the space requirement properties come from the corresponding properties of quantum counting (see Lemma~\ref{lem:q_count_biased}). As the anonymous rule $r$ computes the winner based on an $m!$ dimension histogram, its runtime and space requirement will only depend on the number of candidates $m$, which is $\Theta(1)$ in our analysis. Thus, when $K=1$, Algorithm~\ref{alg:classical_multi} has the following two properties.
\begin{equation}\nonumber
\begin{split}
\text{runtime}
=\;& \underbrace{\Theta(1)}_{\text{the voting rule }r} + \underbrace{\Theta(2^{\tone})}_{\text{quantum counting}} = \Theta(2^{\tone}).
\end{split}
\end{equation}
\begin{equation}\nonumber
\begin{split}
\text{space requirement} =\;&\underbrace{\Theta(1)}_{\text{the voting rule }r} + \underbrace{t+\tone}_{\text{quantum counting}}\\
=\;& \Theta(1)+\lfloor\log(n)\rfloor+s\\
=\;& \Theta(\log(n) + \tone).
\end{split}
\end{equation}

\noindent\textbf{$\bm{\Pr[\text{\textbf{correct}}]}$.} To simplify the notation, let $\thist_j$ be the histogram by adding (or removing) the same fraction of each ranking to $\hhist$ so that the sum of the histogram is $n$: for all dimension $j$,  $\thist_j =\hhist_j +  (||\hist||_1 - ||\hhist||_1)/(m!)$ (recall that $||\hist||_1 = n$). By the assumption of canceling-out (see the second paragraph in Section~\ref{sec:prelim} for its definition), we know that $r(\thist) = r(\hhist)$.


The proof of the $\prc$ property is obtained by the following two propositions. Proposition~\ref{prop:hist_bounded} guarantees that the $\ell_1$-norm between $\hist$ and $\thist$ is bounded by $\MOV$ with at least $1 - \frac{m!}{2(\delta\cdot 2^{\tone} - 1)}$ probability. Proposition~\ref{prop:MOV} guarantees that any histogram in the neighborhood of $\hist$ with $\ell_1$-norm smaller than $2\MOV$ shares the same winner with $\hist$. Therefore, with at least $1 - \frac{m!}{2(\delta\cdot 2^{\tone} - 1)}$ probability, $\thist$ (and therefore $\hhist$) leads to the correct winner. 


The next proposition shows that $\thist_j$ is close to $\hist_j$ with high probability. 


\begin{proposition}
    \label{prop:hist_bounded}
    For all $ \delta \in (2^{-\tone}, \frac{\MOV }{2^{t+1}\pi\cdot(m!)} -2^{-\tone})$, the probability that
    $\lnorm\hist-\thist\rnorm_1 < \MOV$
    is at least $1 - \frac{m!}{2(\delta\cdot 2^{\tone} - 1)}$.  
\end{proposition}


\begin{proof}
    This proof fixes an arbitrary pair of $\tone$ and $\delta$ satisfying the condition. We prove a stronger result that the probability that for all dimension $j$, $|\thist_j - \hist_j| < \frac{\MOV}{m!}$ is at least $1 - \frac{m!}{2(\delta\cdot 2^{\tone} - 1)}$. When the stronger result holds, 
    \begin{equation}\nonumber
    \begin{split}
        \lnorm\hist-\thist\rnorm_1 \le &\ \sum_{j=1}^{m!}|\hist_j - \thist_j| < m!\cdot\frac{\MOV}{m!} = \MOV. 
    \end{split}
    \end{equation}
    Therefore, the probability that for all dimension $j$, $|\thist_j - \hist_j| < \frac{\MOV}{m!}$ is at least $1 - \frac{m!}{2(\delta\cdot 2^{\tone} - 1)}$ directly implies that the probability that
    $\lnorm\hist-\thist\rnorm_1 < \MOV$
    is at least $1 - \frac{m!}{2(\delta\cdot 2^{\tone} - 1)}$.
    
    By applying the union bound on Lemma~\ref{lem:q_count_biased} for all dimensions $j$, we know that for all dimension $j$, the probability that $|\hat{\varphi}_j - \varphi_j| < 2^{-\tone} + \delta$ is at least $1 - \frac{m!}{2(\delta\cdot 2^{\tone} - 1)}$. Then we show that $|\hat{\varphi}_j - \varphi_j| < 2^{-\tone} + \delta$ for all dimension $j$ implies $|\thist_j - \hist_j| < \frac{\MOV}{m!}$ for all $j$. 

    Suppose $|\hat{\varphi}_j - \varphi_j| < 2^{-\tone} + \delta$ holds for all dimension $j$.  Let $h(x) = 2^t\sin^2(\pi x)$. We have $h(\hat{\varphi}_j) = \hhist_j$ and $h({\varphi}_j) = \hist_j$. Note that $\frac{dh(x)}{dx} = 2^t\pi\sin(2\pi x) \le 2^t\pi$ holds for arbitrary $x$. Therefore,  
    \begin{equation}\nonumber
    |\hhist_j - \hist_j| \le 2^t\pi |\hat{\varphi}_j - \varphi_j| < 2^t\pi\cdot(\delta + 2^{-s}).
    \end{equation}

    Let $d = 2^t\pi\cdot(\delta + 2^{-s})$. 
    By summing over all dimensions $j$, we have \begin{equation}\nonumber
        \left|||\hhist ||_1- \lnorm\hist \rnorm_1\right| \le \sum_{j=1} ^{m!} |\hhist_j - \hist_j| < (m!)d.
    \end{equation}
    Given the bound of difference between $\hhist$ and $\hist$, we show that since $\thist$ is the modified $\hhist$, the difference between $\thist$ and $\hist$ is also bounded. Now we consider the difference between $\thist_j$ and $\hist_j$ for any dimension $j$. 

    For the upper bound, we have $\hhist_j < \hist_j + d$ and $(||\hist||_1 - ||\hhist||_1) < (m!)d$. Therefore, 
    \begin{equation}\nonumber
        \begin{split}
            \thist_j = &\; \hhist_j + (||\hist||_1 - ||\hhist||_1)/(m!)\\
            < & \; \hist_j + 2d. 
        \end{split}
    \end{equation}
    Similarly, for the lower bound, we have $\hhist_j > \hist_j - d$. Therefore, $\thist_j > \hist_j - 2d.$

    Therefore, for all dimension $j$, $|\thist_j - \hist_j | < 2d$. The condition that $ \delta < \frac{\MOV}{2^{t+1}\pi\cdot(m!)} -2^{-\tone}$ guarantees that $2d = 2^{t+1}\pi\cdot(\delta + 2^{-\tone}) < \frac{\MOV}{m!}$. Therefore, for all dimension $j$, $|\hist_j - \thist_j| <  \frac{\MOV}{m!}$, which finishes our proof. 
\end{proof}


\begin{proposition}\label{prop:MOV}
    For any histogram $\thist$ such that $\lnorm\hist\rnorm_1 = \lnorm\thist\rnorm_1$ and $\lnorm\hist-\thist\rnorm_1 < 2\MOV$, $r(\hist) = r(\thist)$.
\end{proposition}

\begin{proof}
    We consider the voting rule $r$ that allows voters to vote fractionally. That is, each voter has a total weight of 1 and can assign the weight arbitrarily to every ranking. Accordingly, $\MOV$ is the smallest amount of weight of votes to change the winner. 
    
    Suppose the statement is not true, and there exists a $\thist$ such that $\lnorm\hist-\thist\rnorm_1 < 2\MOV$ and  $r(\hist) \neq r(\thist)$. Let $J_1$ be the set of dimension $j$ such that $\thist_j > \hist _j$, and $J_2$ be the set of $j$ such that $\thist_j < \hist _j$. Since $\lnorm \hist\rnorm_1 = \lnorm \thist\rnorm_1 = n$, there exists a way of transforming $\thist$ to $\hist$ by accumulating all the excess weights of the dimensions in  $J_1$ and assigning it to the dimensions in $J_2$. The changes in the weight \begin{equation}\nonumber
    \sum_{j\in J_1} |\thist_j - \hist _j| = \sum_{j\in J_2} |\thist_j - \hist _j|\ge \MOV
    \end{equation}
    by the definition of $\MOV$. However, this contradicts the assumption that 
    \begin{equation}\nonumber
    \lnorm\hist-\thist\rnorm_1 = \sum_{j\in (J_1\cup J_2)} |\thist_j - \hist _j| < 2\MOV.
    \end{equation}
    Therefore, for any histogram $\thist$ such that $\lnorm\hist-\thist\rnorm_1 < 2\MOV$, $r(\hist) = r(\thist)$.
\end{proof}
Then, the $\prc$ property of Lemma~\ref{thm:hist_same} follows by combining Proposition~\ref{prop:hist_bounded} and Proposition~\ref{prop:MOV}.
\end{proof}
    According to Lemma~\ref{thm:hist_same}, Algorithm~\ref{alg:quantum_multi}'s $\prc \geq 1 - \frac{m!}{2(\delta\cdot 2^{\tone} - 1)}$. In order to achieve $\prc$ of $1 - \varepsilon$, constraint $\delta \ge 2^{-\tone}(\frac{m!}{2\varepsilon}+1)$ needs to be satisfied. Combining the constraint and the feasible region in Lemma~\ref{thm:hist_same}, we know that $\delta\in\left[2^{-\tone}(\frac{m!}{2\varepsilon}+1), \; \frac{\MOV }{2^{t+1}\pi\cdot(m!)} -2^{-\tone}\right)$, which must not be empty, \emph{i.e.}
    \begin{equation}\label{equ:delta}
        2^{-\tone}\left(\frac{m!}{2\varepsilon}+1\right)  <  \frac{\MOV }{2^{t+1}\pi\cdot(m!)} -2^{-\tone}.
    \end{equation}
    It's not hard to verify that $\tone = \sigma(\varepsilon)$ satisfies the constraint in (\ref{equ:delta}). Then, Theorem~\ref{thm:gsr} follows by setting $\delta = 2^{-\tone}(\frac{m!}{2\varepsilon}+1)$ and $\tone = \sigma(\varepsilon)$ for Lemma~\ref{thm:hist_same}.  
\end{proof}

\section{Compare Quantum and Classical Voting}
\label{sec:compare}
This section compares quantum-accelerated voting with (the best performance of) classical voting algorithms. The classical algorithm is designed according to the idea of sampling (either with or without replacement). At the high level, it uses the randomly sampled votes to estimate the winner. We analyze the runtime and space requirements for classical sampling algorithms and compare them with those of our quantum-accelerated voting algorithm. 

\noindent\textbf{When does quantum (may) accelerate voting?}
Firstly, we provide an intuitive explanation of when quantum would accelerate voting the most. 
We first think about the cases where classical algorithms (\emph{e.g.}, randomly sampling a subset of votes and using the subset to predict the winner) do not need to be improved or cannot be improved. When the margin of victory $\MOV = \Theta(n)$, classical algorithms are already very fast according to the Chernoff bound, which says the classical algorithms' error rate can be exponentially small in terms of runtime~\citep{Bhattacharyya2021sample}. Another case is when $\MOV$ is very small (\emph{e.g.}, $\MOV = \Theta(1)$) where classical algorithms' performance is close to the optimal. In this case, any algorithm has to look into each vote to decide the winner. Since the complexity of counting every vote is $\Theta(n)$, there is not a lot of space for the classical algorithms to be improved.  

Between these two extremes is the case where quantum accelerates voting most significantly, for example,  when the margin of victory $\MOV = \Theta(n^{c})$, where $c\in(0,1)$ is a constant. In this case, the classical voting would be as slow as $\Omega(\frac{n^2}{\MOV^2})$ for a fixed error rate $\epsilon$. 
On the other hand, the runtime of the quantum-accelerated voting, $\Theta(\frac{n}{\MOV})$, is quadratically faster. This comparison is also shown in our experiment of $m=2$ and $\MOV = 1024$ for plurality (the middle column in Figure~\ref{fig:plu_m_2}), where the number of voters $n\approx 10^6$ and $\MOV = \sqrt{n}$, \emph{i.e.} the winner gets $\sim$$0.2\%$ more votes than the loser.

Theorem~\ref{thm:classical} establishes a theoretical ``complexity lower bound'' 
of any classical voting algorithms (based on sampling with replacement) for many common voting rules. Consequently,  sampling-based algorithm with replacement is  at least quadratically slower than quantum-accelerated voting. Here, an algorithm being `` sampling-based'' means the sampling method is the only method for the algorithm to get information about the voting profile $P$. We say one voting rule reduces to majority voting for two candidates if the voting rule always has the same winner as majority (using whatever tie-breaking method) when $m=2$. Most commonly used voting rules reduce to majority for two candidates (\emph{e.g.,} any positional scoring rules, STV, Copeland, and maximin, just to name a few). 


\begin{theorem}
    \label{thm:classical}
    Given any fixed $m \ge 2$, for any $\varepsilon\in (0, 0.5]$, any fast voting algorithm based on sampling with replacement for the voting rule such that reduces to majority voting for two candidates 
    requires at least $\Omega\left(\frac{n^2\cdot\left(\frac{1}{2}-\varepsilon\right)^2}{\MOV^2}\right)$ (expected) runtime and at least $\Omega\left(\log\big(\frac{n^2\cdot\left(\frac{1}{2}-\varepsilon\right)^2}{\MOV^2}\big)\right)$ (expected) space to achieve $\prc\geq 1-\varepsilon$ in the worst case. 
\end{theorem}


When $\varepsilon$ is a constant, the bound in Theorem~\ref{thm:classical} becomes $\Omega\left(\frac{n^2}{\MOV^2}\right)$, which is tight according to \citet{Bhattacharyya2021sample}. We note Theorem~\ref{thm:classical} is more general than the bounds in \citet{Bar-Yossef2001:sampling} (require a lower bound for $\varepsilon$) and \citet{Canetti1995sample} (only holds for small-scale samplings). 

Appendix~\ref{app:with_witout} shows that sampling without replacement has the same asymptotic manner as sampling with replacement when the number of samples $T = o(\sqrt{n})$, which is the setting of fast majority voting in many application scenarios.


\begin{proof}[Proof of Theorem~\ref{thm:classical}]
    \textbf{Step 1. } A lower bound of runtime and space requirement for any fast majority voting algorithm for two candidates (Lemma~\ref{lem:classical_m2}).

    \begin{lemma}\label{lem:classical_m2}
    For any $\varepsilon\in(0,0.5]$, any fast (2-candidate) majority voting algorithm based on sampling with replacement requires at least $\Omega\left(\frac{n^2\cdot\left(\frac{1}{2}-\varepsilon\right)^2}{\MOV^2}\right)$ (expected) runtime and at least $\Omega\left(\log\big(\frac{n^2\cdot\left(\frac{1}{2}-\varepsilon\right)^2}{\MOV^2}\big)\right)$ (expected) space to achieve $\prc\geq 1-\varepsilon$.
    \end{lemma}
    The proof of Lemma~\ref{lem:classical_m2} can be found in Appendix~\ref{app:proof}. Note that the $\varepsilon$ part of the bound in Lemma~\ref{lem:classical_m2} is not tight when $\varepsilon$ is close to $0$. This is because our bound is given by the information limit, which cannot be reached when reconstructing a limited number of bits. In particular, our problem only reconstructs 1-bit information (which candidate wins out of the two candidates).
    
    \textbf{Step 2.}  The lower bound for the fast voting algorithm for $m>2$ candidates cannot be smaller than the lower bound for the two-candidate case. 
    
    Suppose Theorem~\ref{thm:classical} is not true, and there exists a fast voting algorithm $A$ for voting rule $r$ such that for any profile $P$, $A$ has a runtime of $o\left(\frac{n^2\cdot\left(\frac{1}{2}-\varepsilon\right)^2}{\MOV^2}\right)$ and achieves $\Pr[A(P) \text{ is correct}] \ge 1 - \varepsilon$ (the reasoning for space complexity will be similar). And suppose $A$ takes a profile $P$ of $m$ candidates $a, b, c_3,\cdots, c_m$ as the input and sample votes from $P$ with replacement. We show that there exists a fast algorithm for majority voting $A'$ with the same runtime and accuracy. Let $P'$ be a voting profile for two candidates $a$ and $b$, and let $n_a$ and $n_b$ be the number of votes for $a$ and for $b$ respectively. We construct algorithm $A'$ as follows: 

    $A'$ is almost the same as $A$ except for the different input and sampling. $A'$ takes a profile $P'$ of two candidates $a$ and $b$ as the input. Whenever there needs a sample, $A$ samples a vote from $P'$. If the voter votes for $a$, then $A'$ convert it to $a \succ b \succ c_3\succ \cdots \succ c_m$; and if it votes for $b$, then $A'$ convert it to $b \succ a \succ c_3\succ \cdots \succ c_m$.
    If the winner calculated is neither $a$ nor $b$, $A'$ will set $a$ to the winner. 
    Except for the sampling, $A'$ running on $P'$ is equivalent to $A$ running on the following profile $P$ and set the same winner as $A$ does: there are $n_a$ votes of $a \succ b \succ c_3\succ \cdots \succ c_m$ and $n_b$ votes of $b \succ a \succ c_3\succ \cdots \succ c_m$.

    If $a$ is the winner in $P'$, then $n_a \ge n_b$, and the margin of victory for the majority vote is $\MOV' = \frac12(n_a - n_b)$. Then it is not hard to verify that when $r$ is one of the rules mentioned in the statement, $a$ is also the winner in $P$, and the margin of victory in $P$ is $\MOV = \MOV'$. Therefore, $A$ set the winner as $A$ with probability at least $1 - \varepsilon$ under runtime $o\left(\frac{n^2\cdot\left(\frac{1}{2}-\varepsilon\right)^2}{\MOV^2}\right)$. Then $A'$, with the same operation to $A$, will also set the winner as $A$ with probability at least $1 - \varepsilon$ under the same runtime. Similarly, if $b$ is the winner, $A'$ can output the correct winner with probability at least $1 - \varepsilon$ under the same runtime. Therefore, $A'$ is a fast majority voting algorithm based on sampling with a replacement that can achieve $\Pr[\text{correct}] \ge 1 - \varepsilon$ under runtime $o\left(\frac{n^2\cdot\left(\frac{1}{2}-\varepsilon\right)^2}{\MOV^2}\right)$, which contradicts with Lemma~\ref{lem:classical_m2}. 
\end{proof}



\begin{figure*}[ht]
    \centering
    \includegraphics[width = 0.995\textwidth]{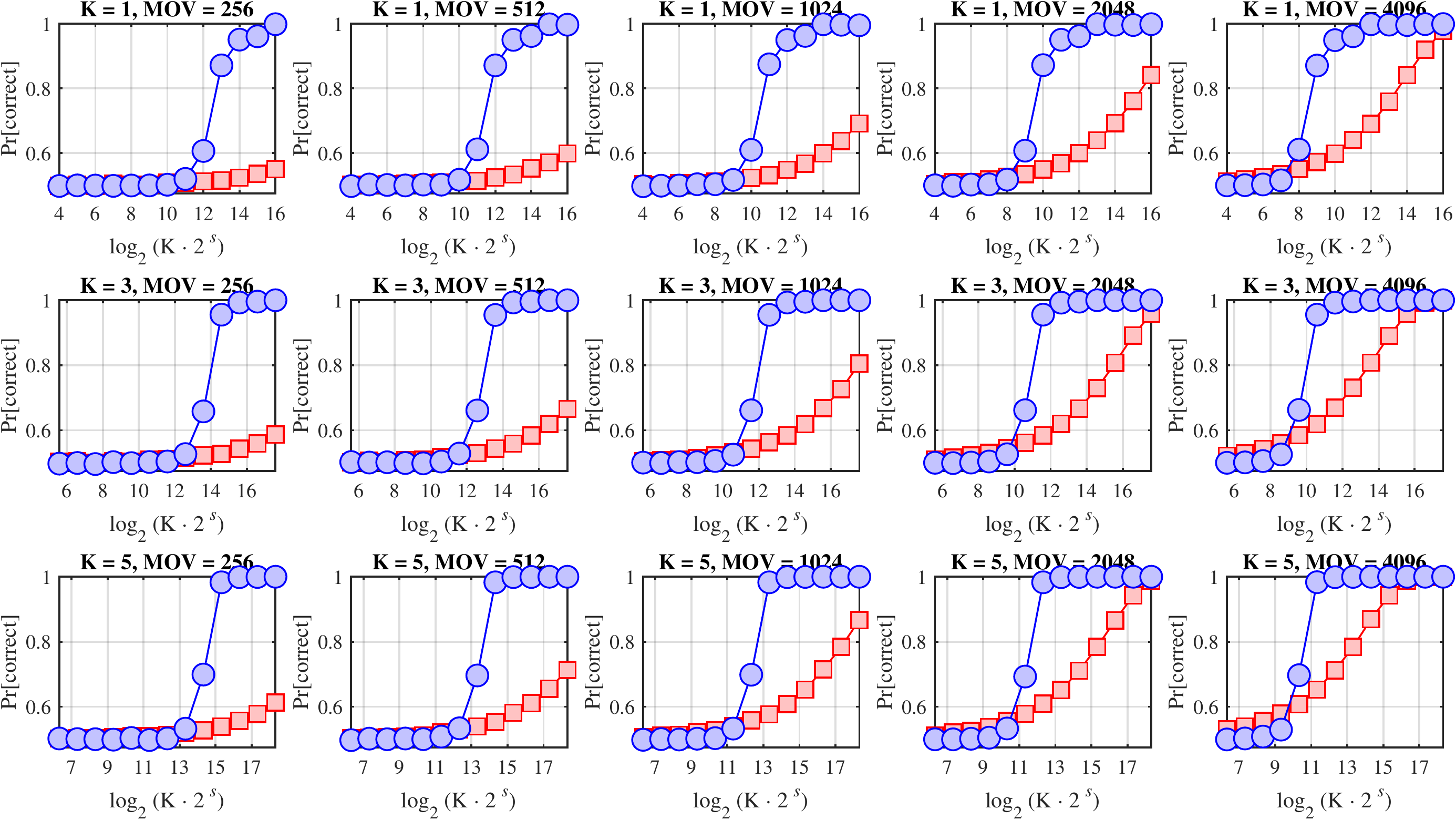}
    \caption{Compare quantum-accelerated voting (blue circles) with classical fast voting (red squares) for plurality, Borda, Copeland, and STV when $m=2$. The horizontal axis can be seen as the logarithm of the algorithms' runtime.}
    \label{fig:plu_m_2}
\end{figure*}

\section{Experimental Results}
\label{sec:exp}
\noindent\textbf{Basic settings. } We numerically compare the proposed quantum-accelerated voting (Algorithm~\ref{alg:quantum_multi}) with a fast classical voting based on sampling with replacement (Algorithm~\ref{alg:classical_multi}). We set the number of samples $T$ in Algorithm~\ref{alg:classical_multi} to be $K\cdot 2^s$, where $s$ and $K$ are the parameters of Algorithm~\ref{alg:quantum_multi}. By doing this, the runtime of both algorithms is $\Theta(K\cdot 2^s)$.  We set the number of voters $n = 2^{20} \approx 10^6$, which is at a similar order of magnitude as the number of voters in a typical state of the United States. For example, the number of registered voters in New Hampshire is 1,009,004 $\approx$ $10^6$~\citep{nh2020}. We compare Algorithm~\ref{alg:quantum_multi} and Algorithm~\ref{alg:classical_multi} on four widely-used voting rules, which are two scoring rules (plurality and Borda), a pairwise rule (Copeland), and an elimination-based rule (STV). The formal definitions of the four rules can be found in Section~\ref{sec:prelim}. Note that all four rules  are covered by Theorem~\ref{thm:classical}. The implementation details can be found in Appendix~\ref{app:detail_exp}. 

\begin{algorithm}[htp]
\caption{Fast Classical Voting Algorithm}\label{alg:classical_multi}
\begin{algorithmic}[1]
\STATE {\bfseries Inputs:} $n$ voters' votes $\vv_0, \cdots, \vv_{n-1}$, a voting rule $r$, number of samples $T$
\STATE Sample $T$ votes uniformly at random with replacement
\STATE Build the histogram $\hhist$ of the sampled votes. 
\STATE {\bfseries Output} $r(\hhist)$ as the winner. 
\end{algorithmic}
\end{algorithm}

\noindent\textbf{Detailed settings. } We use random tie-breaking to break ties for all voting rules above. For example, if $c_1$ and $c_2$ are tied, each of them will win with $1/2$ probability. We set the number of candidates $m \in \{2,4\}$ for all rules.  
In all experiments of this paper, we estimate the probability of outputting the correct winner ($\Pr[\text{correct}]$) by the frequency of observing the correct winner in $10^5$ independent trails.  We set $K \in \{1,3,5\}$ to avoid ties in the classical part of quantum-accelerated voting algorithm. We set $\MOV \in \{256, 512, 1024, 2048, 4096\}$, or equivalently, $\MOV \in \{n^{0.4}, n^{0.45}, n^{0.5}, n^{0.55}, n^{0.6}\}$. Since all four rules reduce to the majority rule when there are only two candidates, we consider the same profile as (\ref{equ:profile_2}) for $m=2$. It's easy to check that the margin of victory of (\ref{equ:profile_2}) is $\MOV$ under all four rules above. Figure~\ref{fig:plu_m_2} compares quantum-accelerated voting with a classical fast voting algorithm (Algorithm~\ref{alg:classical_multi}). The horizontal axis of Figure~\ref{fig:plu_m_2}, $\log_2(K\cdot 2^s)$, can be seen as the logarithm of the algorithms' runtime. For all curves, we set $s = 4,5,\cdots,16$ for the twelve points from left to right, respectively. The detailed settings and experimental results of $m=4$ are shown in Appendix~\ref{app:exp}.

\noindent\textbf{Observations. } First, with the same order of runtime, the quantum-accelerated voting algorithm has better $\prc$ than classical fast voting no matter under which setting. For example, in Figure~\ref{fig:plu_m_2}, $\MOV=1024$, $K=1$, and $s = 14$, the quantum-accelerated voting algorithm outputs the correct winner almost for certain. However, the classical algorithm only has $\sim$$60\%$ probability of outputting the correct winner. Second, quantum-accelerated voting requires much less runtime to achieve the same $\prc$. For example, in Figure~\ref{fig:plu_m_2}, $\MOV=1024$, and $K=1$, to achieve $\sim$$90\%$ $\Pr[\text{correct}]$, the quantum algorithm requires $2^{10}$ runtime. In comparison, the classical algorithm with $2^{16}$ runtime can only achieve $\sim$$70\%$ $\Pr[\text{correct}]$. Both observations match our theoretical result: the proposed algorithm is quadratically faster than any classical algorithm. 

\section{Further accelerating quantum voting}
\label{sec:heuristic}
Can quantum-accelerated voting algorithm be further accelerated? This section discusses some heuristics. 

\noindent\textbf{Pre-sampling.} One way to improve the average performance of the quantum-accelerated voting algorithm is to pre-sample a small subset of the votes. For example, in the majority vote for binary candidates, if the pre-sample votes indicate an almost irreversible win of a candidate, then we directly announce the winner and skip the quantum computing. By carefully setting the pre-sampling size and the skipping thresholds, we may improve the average run-time while keeping high $\prc$. 

\noindent\textbf{Sampling + quantum.} Another natural idea is to apply the quantum-accelerated voting algorithm on a sampled subset of the votes. Sampling decreases the number of votes, so the quantum circuit consumes fewer bits and operators, which reduces the time and space cost. However, such improvement sacrifices $\prc$, as both sampling and quantum computing has a probability to make a mistake. It is still unclear if such a sampling-quantum algorithm would be faster to achieve the same level of $\prc$.

\section{Conclusions and Future Works}
In this paper, we took the first step in using quantum computation to accelerate voting. Our proposed quantum-accelerated voting algorithm can quadratically accelerate various widely-used voting rules and may potentially improve the efficiency of voting in large-scale and/or high-frequency decision-making scenarios. 

An extension of this paper is to further accelerate the proposed algorithm by combining existing acceleration techniques in classical fast voting algorithms. Since voting is widely used in artificial intelligence, it would also be interesting to apply the proposed methods to accelerate the algorithms in other fields (\emph{e.g.,} search engine, crowdsourcing, database management, and blockchain governance). 

As real quantum computers may come across quantum errors caused by quantum interference and environmental effects, another interesting extension is to test the robustness of the proposed algorithm. For example, testing its performance on real quantum computers and/or running experiments with quantum errors taken into account.

\begin{acknowledgements}
We thank the anonymous reviewers for their helpful comments. Lirong Xia is supported by NSF \#1453542 and a gift fund from Google. 
\end{acknowledgements}

\bibliography{sample}

\begin{thebibliography}{46}
\providecommand{\natexlab}[1]{#1}
\providecommand{\url}[1]{\texttt{#1}}
\expandafter\ifx\csname urlstyle\endcsname\relax
  \providecommand{\doi}[1]{doi: #1}\else
  \providecommand{\doi}{doi: \begingroup \urlstyle{rm}\Url}\fi

\bibitem[Ajagekar and You(2020)]{AJAGEKAR2020107119}
Akshay Ajagekar and Fengqi You.
\newblock Quantum computing assisted deep learning for fault detection and
  diagnosis in industrial process systems.
\newblock \emph{Computers \& Chemical Engineering}, 143:\penalty0 107119, 2020.

\bibitem[Ajagekar and You(2021)]{AJAGEKAR2021117628}
Akshay Ajagekar and Fengqi You.
\newblock Quantum computing based hybrid deep learning for fault diagnosis in
  electrical power systems.
\newblock \emph{Applied Energy}, 303:\penalty0 117628, 2021.

\bibitem[Bar-Yossef et~al.(2001)Bar-Yossef, Kumar, and
  Sivakumar]{Bar-Yossef2001:sampling}
Ziv Bar-Yossef, Ravi Kumar, and D.~Sivakumar.
\newblock Sampling algorithms: Lower bounds and applications.
\newblock In \emph{Proceedings of the Thirty-Third Annual ACM Symposium on
  Theory of Computing}, STOC '01, page 266–275, New York, NY, USA, 2001.
  Association for Computing Machinery.
\newblock ISBN 1581133499.

\bibitem[Bartholdi et~al.(1989)Bartholdi, Tovey, and
  Trick]{bartholdi1989voting}
John Bartholdi, Craig~A Tovey, and Michael~A Trick.
\newblock Voting schemes for which it can be difficult to tell who won the
  election.
\newblock \emph{Social Choice and welfare}, 6:\penalty0 157--165, 1989.

\bibitem[Bausch(2020)]{bausch2020recurrent}
Johannes Bausch.
\newblock Recurrent quantum neural networks.
\newblock In \emph{Proceedings of the 34th International Conference on Neural
  Information Processing Systems}, pages 1368--1379, 2020.

\bibitem[Belardinelli and Grandi(2019)]{Belardinelli_2019}
Francesco Belardinelli and Umberto Grandi.
\newblock Social choice methods for database aggregation.
\newblock In \emph{17th Conference on Theoretical Aspects of Rationality and
  Knowledge (TARK 2019)}, volume 297, pages 50--67, 2019.

\bibitem[Benedetti et~al.(2016)Benedetti, Realpe-G{\'o}mez, Biswas, and
  Perdomo-Ortiz]{PhysRevA.94.022308}
Marcello Benedetti, John Realpe-G{\'o}mez, Rupak Biswas, and Alejandro
  Perdomo-Ortiz.
\newblock Estimation of effective temperatures in quantum annealers for
  sampling applications: A case study with possible applications in deep
  learning.
\newblock \emph{Physical Review A}, 94\penalty0 (2):\penalty0 022308, 2016.

\bibitem[Berthiaume and Brassard(1994)]{berthiaume1994oracle}
Andr{\'e} Berthiaume and Gilles Brassard.
\newblock Oracle quantum computing.
\newblock \emph{Journal of modern optics}, 41\penalty0 (12):\penalty0
  2521--2535, 1994.

\bibitem[Bhattacharyya and Dey(2021)]{Bhattacharyya2021sample}
Arnab Bhattacharyya and Palash Dey.
\newblock Predicting winner and estimating margin of victory in elections using
  sampling.
\newblock \emph{Artificial Intelligence}, 296:\penalty0 103476, 2021.

\bibitem[Brandt et~al.(2016)Brandt, Conitzer, Endriss, Lang, and
  Procaccia]{brandt2016handbook}
Felix Brandt, Vincent Conitzer, Ulle Endriss, J{\'e}r{\^o}me Lang, and Ariel~D
  Procaccia.
\newblock \emph{Handbook of computational social choice}.
\newblock Cambridge University Press, 2016.

\bibitem[Brassard et~al.(1998)Brassard, H{\o}yer, and
  Tapp]{Brassard1998counting}
Gilles Brassard, Peter H{\o}yer, and Alain Tapp.
\newblock Quantum counting.
\newblock In \emph{Automata, Languages and Programming: 25th International
  Colloquium, ICALP'98 Aalborg, Denmark, July 13--17, 1998 Proceedings 25},
  pages 820--831. Springer, 1998.

\bibitem[Brill(2018)]{brill2018interactive}
Markus Brill.
\newblock Interactive democracy.
\newblock In \emph{Proceedings of the 17th International Conference on
  Autonomous Agents and MultiAgent Systems}, pages 1183--1187, 2018.

\bibitem[Canetti et~al.(1995)Canetti, Even, and Goldreich]{Canetti1995sample}
Ran Canetti, Guy Even, and Oded Goldreich.
\newblock Lower bounds for sampling algorithms for estimating the average.
\newblock \emph{Information Processing Letters}, 53\penalty0 (1):\penalty0
  17--25, 1995.

\bibitem[Chen et~al.(2020)Chen, Wei, Gao, Wang, Tang, Wu, and Guo]{chen2020low}
Yanhu Chen, Shijie Wei, Xiong Gao, Cen Wang, Yinan Tang, Jian Wu, and Hongxiang
  Guo.
\newblock A low failure rate quantum algorithm for searching maximum or
  minimum.
\newblock \emph{Quantum Information Processing}, 19\penalty0 (8), 2020.

\bibitem[Conitzer et~al.(2009)Conitzer, Rognlie, and
  Xia]{conitzer2009preference}
Vincent Conitzer, Matthew Rognlie, and Lirong Xia.
\newblock Preference functions that score rankings and maximum likelihood
  estimation.
\newblock In \emph{Twenty-First International Joint Conference on Artificial
  Intelligence}, 2009.

\bibitem[Dwork et~al.(2001)Dwork, Kumar, Naor, and Sivakumar]{dwork2001rank}
Cynthia Dwork, Ravi Kumar, Moni Naor, and Dandapani Sivakumar.
\newblock Rank aggregation methods for the web.
\newblock In \emph{Proceedings of the 10th international conference on World
  Wide Web}, pages 613--622, 2001.

\bibitem[Flanigan et~al.(2020)Flanigan, G{\"o}lz, Gupta, and
  Procaccia]{Flanigan2020sortition}
Bailey Flanigan, Paul G{\"o}lz, Anupam Gupta, and Ariel~D Procaccia.
\newblock Neutralizing self-selection bias in sampling for sortition.
\newblock In \emph{Proceedings of the 34th International Conference on Neural
  Information Processing Systems}, pages 6528--6539, 2020.

\bibitem[Flanigan et~al.(2021)Flanigan, G{\"o}lz, Gupta, Hennig, and
  Procaccia]{flanigan2021fair}
Bailey Flanigan, Paul G{\"o}lz, Anupam Gupta, Brett Hennig, and Ariel~D
  Procaccia.
\newblock Fair algorithms for selecting citizens’ assemblies.
\newblock \emph{Nature}, 596\penalty0 (7873):\penalty0 548--552, 2021.

\bibitem[Giovannetti et~al.(2008{\natexlab{a}})Giovannetti, Lloyd, and
  Maccone]{giovannetti2008arch}
Vittorio Giovannetti, Seth Lloyd, and Lorenzo Maccone.
\newblock Architectures for a quantum random access memory.
\newblock \emph{Physical Review A}, 78\penalty0 (5):\penalty0 052310,
  2008{\natexlab{a}}.

\bibitem[Giovannetti et~al.(2008{\natexlab{b}})Giovannetti, Lloyd, and
  Maccone]{giovannetti2008quantum}
Vittorio Giovannetti, Seth Lloyd, and Lorenzo Maccone.
\newblock Quantum random access memory.
\newblock \emph{Physical review letters}, 100\penalty0 (16):\penalty0 160501,
  2008{\natexlab{b}}.

\bibitem[Grossi(2022)]{grossi2022social}
Davide Grossi.
\newblock Social choice around the block: On the computational social choice of
  blockchain.
\newblock In \emph{Proceedings of the 21st International Conference on
  Autonomous Agents and Multiagent Systems}, pages 1788--1793, 2022.

\bibitem[Grover(1996)]{10.1145/237814.237866}
Lov~K Grover.
\newblock A fast quantum mechanical algorithm for database search.
\newblock In \emph{Proceedings of the twenty-eighth annual ACM symposium on
  Theory of computing}, pages 212--219, 1996.

\bibitem[Hogg and Portnov(2000)]{hogg2000quantum}
Tad Hogg and Dmitriy Portnov.
\newblock Quantum optimization.
\newblock \emph{Information Sciences}, 128\penalty0 (3-4):\penalty0 181--197,
  2000.

\bibitem[Homan and Hemaspaandra(2009)]{homan2009guarantees}
Christopher~M Homan and Lane~A Hemaspaandra.
\newblock Guarantees for the success frequency of an algorithm for finding
  dodgson-election winners.
\newblock \emph{Journal of Heuristics}, 15:\penalty0 403--423, 2009.

\bibitem[Independent Voter~Project(2020)]{nh2020}
2020 Independent Voter~Project.
\newblock New hampshire voter statistics, 2020.

\bibitem[Kashefi et~al.(2002)Kashefi, Kent, Vedral, and
  Banaszek]{kashefi2002comparison}
Elham Kashefi, Adrian Kent, Vlatko Vedral, and Konrad Banaszek.
\newblock Comparison of quantum oracles.
\newblock \emph{Physical Review A}, 65\penalty0 (5):\penalty0 050304, 2002.

\bibitem[Kay(2018)]{kay2018tutorial}
Alastair Kay.
\newblock Tutorial on the quantikz package, 2018.

\bibitem[Kerenidis et~al.(2019)Kerenidis, Landman, Luongo, and
  Prakash]{kerenidis2019q}
Iordanis Kerenidis, Jonas Landman, Alessandro Luongo, and Anupam Prakash.
\newblock q-means: a quantum algorithm for unsupervised machine learning.
\newblock In \emph{Proceedings of the 33rd International Conference on Neural
  Information Processing Systems}, pages 4134--4144, 2019.

\bibitem[Khabiboulline et~al.(2021)Khabiboulline, Sandhu, Gambetta, Lukin, and
  Borregaard]{khabiboulline2021efficient}
Emil~T Khabiboulline, Juspreet~Singh Sandhu, Marco~Ugo Gambetta, Mikhail~D
  Lukin, and Johannes Borregaard.
\newblock Efficient quantum voting with information-theoretic security, 2021.

\bibitem[Lian and Zhang(2009)]{lian2009handbook}
Shiguo Lian and Yan Zhang.
\newblock \emph{Handbook of research on secure multimedia distribution}.
\newblock IGI Global, Hershey, PA, 2009.
\newblock ISBN 1605662623.

\bibitem[Liu et~al.(2020)Liu, Lu, Xia, and Zikas]{ao2020private}
Ao~Liu, Yun Lu, Lirong Xia, and Vassilis Zikas.
\newblock How private are commonly-used voting rules?
\newblock In \emph{Conference on Uncertainty in Artificial Intelligence}, pages
  629--638. PMLR, 2020.

\bibitem[MacKay(2003)]{mackay2003information}
David~JC MacKay.
\newblock \emph{Information theory, inference and learning algorithms}.
\newblock Cambridge university press, 2003.

\bibitem[Mancini(2015)]{mancini2015time}
Pia Mancini.
\newblock Why it is time to redesign our political system.
\newblock \emph{European View}, 14\penalty0 (1):\penalty0 69--75, 2015.

\bibitem[Mao et~al.(2013)Mao, Procaccia, and Chen]{mao2013better}
Andrew Mao, Ariel Procaccia, and Yiling Chen.
\newblock Better human computation through principled voting.
\newblock In \emph{Proceedings of the AAAI Conference on Artificial
  Intelligence}, volume~27, pages 1142--1148, 2013.

\bibitem[McCabe-Dansted et~al.(2008)McCabe-Dansted, Pritchard, and
  Slinko]{mccabe2008approximability}
John~C McCabe-Dansted, Geoffrey Pritchard, and Arkadii Slinko.
\newblock Approximability of dodgson’s rule.
\newblock \emph{Social Choice and Welfare}, 31\penalty0 (2):\penalty0 311--330,
  2008.

\bibitem[Nielsen and Chuang(2010)]{nielsen2002quantum}
Michael~A. Nielsen and Isaac~L. Chuang.
\newblock \emph{Quantum Computation and Quantum Information: 10th Anniversary
  Edition}.
\newblock Cambridge University Press, 2010.

\bibitem[Park et~al.(2019)Park, Petruccione, and Rhee]{park2019circuit}
Daniel~K Park, Francesco Petruccione, and June-Koo~Kevin Rhee.
\newblock Circuit-based quantum random access memory for classical data.
\newblock \emph{Scientific reports}, 9\penalty0 (1):\penalty0 3949, 2019.

\bibitem[Shannon(1948)]{shannon1948mathematical}
Claude~E Shannon.
\newblock A mathematical theory of communication.
\newblock \emph{The Bell system technical journal}, 27\penalty0 (3):\penalty0
  379--423, 1948.

\bibitem[Teerapabolarn and Wongkasem(2011)]{Teerapabolarn2011:pointwise}
Kanint Teerapabolarn and Patcharee Wongkasem.
\newblock On pointwise binomial approximation by w-functions.
\newblock \emph{International Journal of Pure and Applied Mathematics},
  71:\penalty0 57--66, 01 2011.

\bibitem[Tops{\o}e(2001)]{topsoe2001bounds}
Flemming Tops{\o}e.
\newblock Bounds for entropy and divergence for distributions over a
  two-element set.
\newblock \emph{J. Ineq. Pure \& Appl. Math}, 2\penalty0 (2):\penalty0 300,
  2001.

\bibitem[Vaccaro et~al.(2007)Vaccaro, Spring, and Chefles]{vaccaro2007quantum}
Joan~Alfina Vaccaro, Joseph Spring, and Anthony Chefles.
\newblock Quantum protocols for anonymous voting and surveying.
\newblock \emph{Physical Review A}, 75\penalty0 (1):\penalty0 012333, 2007.

\bibitem[Van~Dam(1998)]{van1998quantum}
Wim Van~Dam.
\newblock Quantum oracle interrogation: Getting all information for almost half
  the price.
\newblock In \emph{Proceedings 39th Annual Symposium on Foundations of Computer
  Science (Cat. No. 98CB36280)}, pages 362--367. IEEE, 1998.

\bibitem[Walsh and Xia(2012)]{walsh2012lot}
Toby Walsh and Lirong Xia.
\newblock Lot-based voting rules.
\newblock In \emph{Proceedings of the 11th International Conference on
  Autonomous Agents and Multiagent Systems-Volume 2}, pages 603--610, 2012.

\bibitem[Wang et~al.(2019)Wang, Sikdar, Shepherd, Zhao, Jiang, and
  Xia]{Wang2019PUT}
Jun Wang, Sujoy Sikdar, Tyler Shepherd, Zhibing Zhao, Chunheng Jiang, and
  Lirong Xia.
\newblock Practical algorithms for multi-stage voting rules with parallel
  universes tiebreaking.
\newblock In \emph{Proceedings of the AAAI Conference on Artificial
  Intelligence}, volume~33, pages 2189--2196, 2019.

\bibitem[Xia and Zheng(2022)]{xia2022beyond}
Lirong Xia and Weiqiang Zheng.
\newblock Beyond the worst case: Semi-random complexity analysis of winner
  determination.
\newblock In \emph{Web and Internet Economics: 18th International Conference,
  WINE 2022, Troy, NY, USA, December 12--15, 2022}, pages 330--347. Springer,
  2022.

\bibitem[Xue and Zhang(2017)]{xue2017simple}
Peng Xue and Xin Zhang.
\newblock A simple quantum voting scheme with multi-qubit entanglement.
\newblock \emph{Scientific reports}, 7\penalty0 (1):\penalty0 1--4, 2017.

\end{thebibliography}




\clearpage
\begin{center}
    \textbf{\large The Appendix of UAI-23 Accepted Paper}\\ 
    \textbf{\large Accelerating Voting by Quantum Computation}
\end{center}
\appendix
\section{Implementation of Quantum Counting Algorithm.}
\label{apx:quantum}
In this section, we aim to introduce the implementation of the quantum part of Algorithm~\ref{alg:quantum_multi} from a more technical perspective. We will first introduce the basics of quantum computing. Then we will specify the implementation of circuits of quantum counting in Algorithm~\ref{alg:quantum_multi}, and why they accelerate the voting process. 

\subsection{Quantum Basics.}
\label{apx:basic}
\noindent\textbf{Basic quantum computation.}
Quantum bit (or \emph{qubit} in short) is the counterpart of classical \emph{bit}, which takes a deterministic binary from $\{0,1\}$. Qubit, on the other hand, is represented by a linear combination of $\{\zero,\one\}$, which are counterparts to $\{0,1\}$, respectively. That is, every qubit  $|\psi\rangle$ is written as 
\begin{equation}\nonumber
|\psi\rangle = \alpha\zero + \beta\one,    
\end{equation}
where $\alpha$ and $\beta$ are complex numbers and are usually called amplitudes. If we measure the qubit, there is $|\alpha|^2$ probability to get $0$ and $|\beta|^2$ probability to get $1$. Naturally, we always have $|\alpha|^2+|\beta|^2 = 1$ because the probabilities should sum to $1$. Qubits sometimes are written as vectors to simplify notations. Formally, 
\begin{equation}\nonumber
\begin{bmatrix}
\alpha\\
\beta
\end{bmatrix} \triangleq \alpha\zero + \beta\one.   
\end{equation}
$t > 1$ qubits are presented as a $2^t$-dimensional vector, where the $j$-th component of the vector (denoted as $\alpha_j$) represents the amplitude of $|j_1\cdots j_t\rangle$ (or $|j\rangle$), where $j_1\cdots j_t$ is the binary representation of $j$. Similar to the 1-qubit case, the probability of  observing $j_1,\cdots,j_t$ from those $t$ qubit equals to $|\alpha_j|^2$. 

A quantum operation (quantum gate) $Q$ on $t$ qubits is denoted by a $2^t\times 2^t$ unitary matrix, which means the matrix's inverse is its Hermitian conjugate. Applying a quantum operation $Q$ on quantum state $|\psi\rangle$ is denoted by
\begin{equation}\nonumber
Q|\psi\rangle \triangleq \boldsymbol{Q}_{(2^t\times 2^t)} \; \vec{\psi}_{(2^t)},
\end{equation}
where the the quantum operator $\boldsymbol{Q}_{(2^t\times 2^t)}$ is a $2^t\times 2^t$ unitary matrix and the quantum state $\vec{\psi}_{(2^t)}$ is a $2^t$ dimensional column vector.\\

\noindent\textbf{Quantum circuit of some useful quantum operators.\footnote{All quantum circuits of this paper are drawn using the Quantikz package~\citep{kay2018tutorial} for \LaTeX.}  } Quantum circuits run from the left-hand side to the right-hand side. For example, the following circuit means applying Hadamard gate $H$ on a quantum state $|\psi\rangle$.
\begin{equation}\nonumber
\begin{quantikz} 
& \ket{\psi}\; & \gate{H} & \qw 
\end{quantikz} \qquad \text{ where } \boldsymbol{H} = \frac{1}{\sqrt 2} \begin{bmatrix}
1 & 1\\
1 & -1
\end{bmatrix}.
\end{equation}

The quantum circuit notion
\begin{equation}\nonumber
\begin{quantikz} 
\ket{\psi}\; \qw & \meter{0/1} &  \measuretab{b}
\end{quantikz}
\end{equation}
denotes measuring quantum state $|\psi\rangle$ with $0/1$ base ($b$ denotes the result of measurement). Naturally, the complexity of quantum measurement and Hadamard gate are both $\Theta(1)$.

Quantum oracle~\citep{berthiaume1994oracle,van1998quantum,kashefi2002comparison} is a widely-used operator to encode binary functions or binary information. Given $t$ qubits and a binary function $f:\{0,\cdots,2^{t}-1\}\mapsto \{0,1\}$, quantum oracle (based on function $f(\cdot)$) applies a phase shift of $-1 = e^{\pi i}$ if $f(x) = 1$ and does nothing otherwise. We can query oracle many times and regard the number of queries as the cost \citep{10.1145/237814.237866}. Formally,
\begin{equation}\nonumber
\left\{
\begin{array}{ll}
O_f\x = \x  & \text{if } f(x) = 1 \\
O_f\x = -\x & \text{otherwise}
\end{array}
\right..
\end{equation}

Suppose we have a quantum gate $G$ on $t$ qubits. The following operation is called \emph{controlled-$G$}.
\begin{equation}\nonumber
\begin{quantikz} 
\qw & \ctrl{1} & \qw \\
\qw & \gate{G} & \qw 
\end{quantikz} = \begin{bmatrix}
\boldsymbol{I}_{(2^t\times 2^t)} & \boldsymbol{0}_{(2^t\times 2^t)}\\
\boldsymbol{0}_{(2^t\times 2^t)} & \boldsymbol{G}_{(2^t\times 2^t)}
\end{bmatrix},
\end{equation}
where $\boldsymbol{I}$ denotes the identity matrix, and $\boldsymbol{0}$ denotes the zeros matrix. 
To simplify notations, we also write
\begin{equation}\nonumber
\begin{quantikz} 
& \ctrl{1} & \qw \\
& \gate{G^a} & \qw 
\end{quantikz} = \begin{quantikz} 
& \ctrl{1} & \qw \ \ldots \ \qw & \ctrl{1} & \qw\\
& \gate{G} & \qw \ \ldots \ \qw & \gate{G} & \qw
\end{quantikz}\;{\text{(repeat } a \text{ times)} }.
\end{equation}

\subsection{Implementation of Quantum Counting Circuit.}
\label{apx:circuit}
Figure~\ref{fig:qc2} shows the quantum counting circuit, which is a combination of Grover search algorithm~\citep{10.1145/237814.237866} and quantum reverse
Fourier transformation (the $QFT^{\dagger}$ operator) Followings we focus on introducing Grover algorithm and why it accelerates the computation. 

\begin{figure*}[htp]
\begin{quantikz}
\lstick[wires=4]{Register 1\\$\tone$ qubits} & \ket{0}\;\, &  \gate{H} & \qw & \ctrl{4}& \qw & \qw \  \ldots \ & \qw & \gate[4]{QFT^{\dagger}} & \meter{0/1} & \measuretab{b_{1}}\\
& \ket{0}\;\, & \gate{H} & \qw & \qw  & \ctrl{3} & \qw  \ \ldots \ & \qw & \qw & \meter{0/1} &  \measuretab{b_{2}}\\
& \vdots & \vdots & & &   & \vdots &  & \qwbundle[alternate]{} & \qwbundle[alternate]{} & \vdots \\
& \ket{0}\;\, & \gate{H} &  \qw & \qw & \qw  & \qw  \ \ldots \ & \ctrl{1} & \qw &  \meter{0/1} &  \measuretab{b_{\tone}}\\
\lstick[wires=3]{Register 2\\$t$ qubits} & \ket{0}\;\, & \gate{H} & \qw & \gate[3,bundle={2}]{G^{2^{0}}}  & \gate[3,bundle={2}]{G^{2^{1}}}  & \qw \ \ldots \ & \gate[3,bundle={2}]{G^{2^{\tone-1}}}  & \qw &  \rstick[wires=3]{trash}\\
& \vdots & \vdots & & &   & \qwbundle[alternate]{} \,\ldots\, & & \qwbundle[alternate]{} \\
& \ket{0}\;\, & \gate{H} & \qw &   &  & \qw \ \ldots \ &   & \qw &  \\
\end{quantikz}
\caption{The circuit for quantum counting algorithm.}\label{fig:qc2}
\end{figure*}
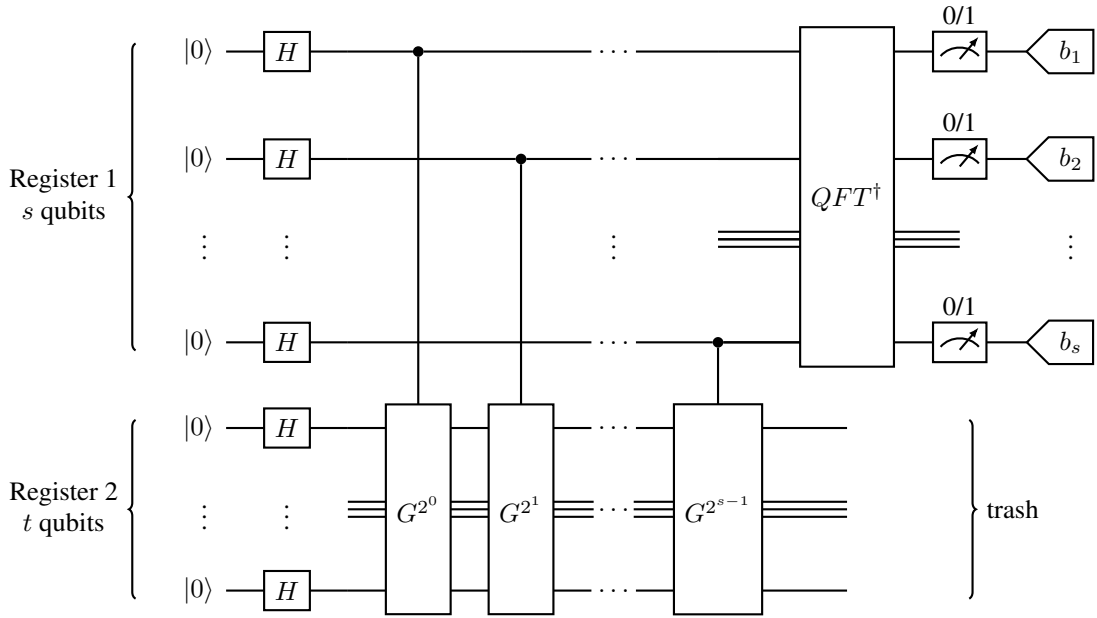

\begin{figure}[htp]
\begin{quantikz}
\lstick[wires=3]{$t$ qubits}   & \gate[3,bundle={2}]{O_{f_j}} & \gate{H} & \gate[3,bundle={2}]{\qps} & \gate{H}& \qw\\ 
 &  & \qwbundle[alternate]{} \;\; \vdots\;\;  & & \qwbundle[alternate]{} \;\;\vdots \;\; & \qwbundle[alternate]{} \\
  & & \gate{H} & & \gate{H} & \qw\\
\end{quantikz}
\caption{The circuit for Grover operator.}\label{fig:qc1}
\end{figure}
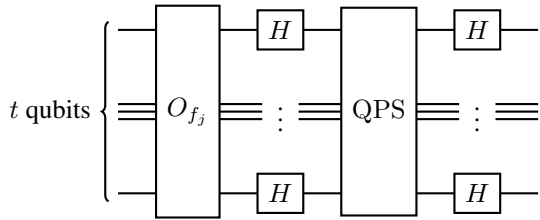

\noindent\textbf{Grover operator.} Grover algorithm is an efficient search algorithm. Given a binary function $f: \{0,1,\cdots, 2^t - 1\} \to \{0, 1\}$, Grover algorithms returns an $x$ with $f(x) = 1$ with high probability. The Grover operation in Algorithm~\ref{alg:quantum_multi} is constructed by the quantum circuit in Figure~\ref{fig:qc1}, where $t = \lceil\log n\rceil$ denotes the minimum number of quantum bits to encode $n$.
The quantum operator $\qps$ is called quantum phase shifting, which provides a phase shift of $-1$ on every state except $|0\rangle$. Mathematically,
\begin{equation}\nonumber
\begin{split}
& \zero \stackrel{\qps}{\longrightarrow} \zero\qquad\text{ and }\\
& \x \stackrel{\qps}{\longrightarrow} -\x\; \text{ for any } x \in{1,\cdots,2^{t}-1}.
\end{split}
\end{equation}
Here, $\x$ represents the $x$-th base state of the $t$ qubits. The high-level idea of Grover operator's functionality is shown in Figure~\ref{fig:grover}, where $|\psi\rangle$ is the input of Grover operators in quantum counting, and $\{|\alpha\rangle, |\beta\rangle\}$ is a pair of orthogonal bases. The formal definition of $|\psi\rangle$, $|\alpha\rangle$, and $|\beta\rangle$ can be found in Appendix~\ref{app:add}. Under the $|\alpha\rangle$ $|\beta\rangle$ base, the quantum oracle $O_{f_j}$ reflects $|\psi\rangle$ over $|\alpha\rangle$, while the rest parts of $G$ reflects $O_{f_j}|\psi\rangle$ over $|\psi\rangle$. The angle between the output state $G|\psi\rangle$ and initial state $|\psi\rangle$
\begin{equation}\nonumber
\theta = 2\arcsin\left(\sqrt{\hist_j\cdot 2^{-t}}\right),
\end{equation}
which includes the information about $\hist_j$. Since function $\arcsin(\sqrt{x})$ grows quadratically faster than linear functions when $x$ is small, we expect that an estimation about $\arcsin(\sqrt{x})$ could be quadratically more accurate than directly estimate $x$. \\

\begin{figure}[htp]
    \centering
    \includegraphics[width = 0.378\textwidth]{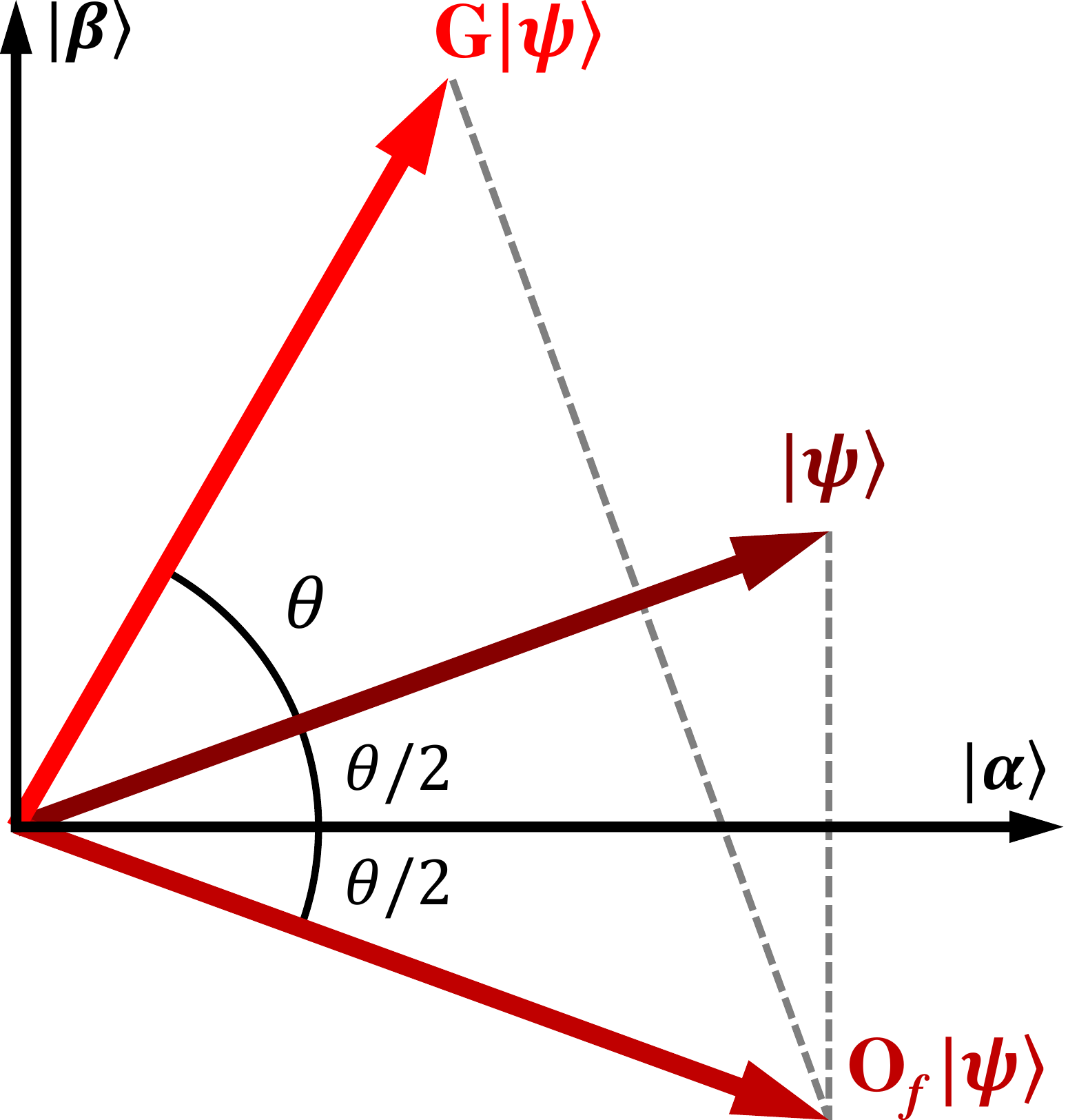}
    \caption{An illustration of Grover operator's functionality (Figure 6.3 in~\citet{nielsen2002quantum}).}
    \label{fig:grover}
\end{figure}

\subsection{Functionality for Grover algorithm}\label{app:add}
According to (6.4) in~\citet{nielsen2002quantum}, Hadamard gate changes $t$ qubits of $\zero$ to an equal superposition state (equal probability of observing any outcome under quantum measurements).
\begin{equation}\nonumber
|\psi\rangle = \frac{1}{2^{t/2}}\cdot \sum_{x=0}^{2^t-1} |x\rangle.
\end{equation}
Letting $f:\{0,\cdots,2^{t}-1\}\mapsto \{0,1\}$ be the binary function to construct the quantum oracle, and $\hn_1$ be the number of $x$ such that $f(x) =1$. The orthogonal bases $|\alpha\rangle$ and $|\beta\rangle$ are defined as,
\begin{equation}\nonumber
\begin{split}
|\alpha\rangle &\triangleq \frac{1}{\sqrt{2^t-\hn_1}}\cdot \sum_{x:f(x) = 0} |x\rangle\qquad\text{and}\\
|\beta\rangle &\triangleq \frac{1}{\sqrt{\hn_1}}\cdot \sum_{x:f(x) = 1} |x\rangle.
\end{split}    
\end{equation}
Under the $|\alpha\rangle$ $|\beta\rangle$ base,  the equal superposition state
\begin{equation}\nonumber
|\psi\rangle = \sqrt{\frac{2^t-\hn_1}{2^t}} 
\;|\alpha\rangle + \sqrt{\frac{\hn_1}{2^t}}\;|\beta\rangle.
\end{equation}
Since 
\begin{equation}\nonumber
\theta = 2\arcsin\left(\sqrt{\hn_1\cdot 2^{-t}}\right),
\end{equation}
we have
\begin{equation}\nonumber
\begin{split}
|\psi\rangle &= \cos\left(\frac{\theta}{2}\right)
\,|\alpha\rangle + \sin\left(\frac{\theta}{2}\right)\,|\beta\rangle, \\
O_f|\psi\rangle &= \cos\left(\frac{\theta}{2}\right)
\,|\alpha\rangle + \sin\left(-\frac{\theta}{2}\right)\,|\beta\rangle, \text{ and}\\
G|\psi\rangle &= \cos\left(\frac{3\theta}{2}\right)
\,|\alpha\rangle + \sin\left(\frac{3\theta}{2}\right)\,|\beta\rangle.
\end{split}    
\end{equation}



\section{Missing proofs and discussions}
\subsection{Missing Proof for Lemma~\ref{lem:classical_m2}}\label{app:proof}

\textbf{Lemma~\ref{lem:classical_m2}.}
\emph{Given $\varepsilon\in(0,0.5]$, any fast (2-candidate) majority voting algorithm based on sampling with replacement requires at least $\Omega\left(\frac{n^2\cdot\left(\frac{1}{2}-\varepsilon\right)^2}{\MOV^2}\right)$ runtime and at least $\Omega\left(\log\big(\frac{n^2\cdot\left(\frac{1}{2}-\varepsilon\right)^2}{\MOV^2}\big)\right)$ space to achieve $\prc\geq 1-\varepsilon$.}
\begin{proof}
For majority voting (when $m=2$), the corresponding profile with margin of victory $\MOV$ is 
\begin{equation}\label{equ:profile_2}
\left\{
\begin{array}{l}
\nwi = (\lfloor n/2\rfloor+\MOV) \text{ votes for the winner}\\
\\
\nlo = (\lceil n/2\rceil-\MOV) \text{ votes for the loser}
\end{array}
\right..
\end{equation}
Figure~\ref{fig:channel} interprets the sampling (with replacement) process as a communication problem. We (the receiver) get a noisy data point about the winner from the sampling process. According to the above profile, we get the correct winner with $\frac{\nwi}{n}$ probability and get the incorrect winner with $\frac{\nlo}{n} = 1-\frac{\nwi}{n}$ probability. This sampling process is equivalent to the noisy communication channel in Figure~\ref{fig:channel}, which gives the correct binary message with $\frac{\nwi}{n}$ probability.

\begin{figure}[ht]
    \centering
    \includegraphics[width = 0.48\textwidth]{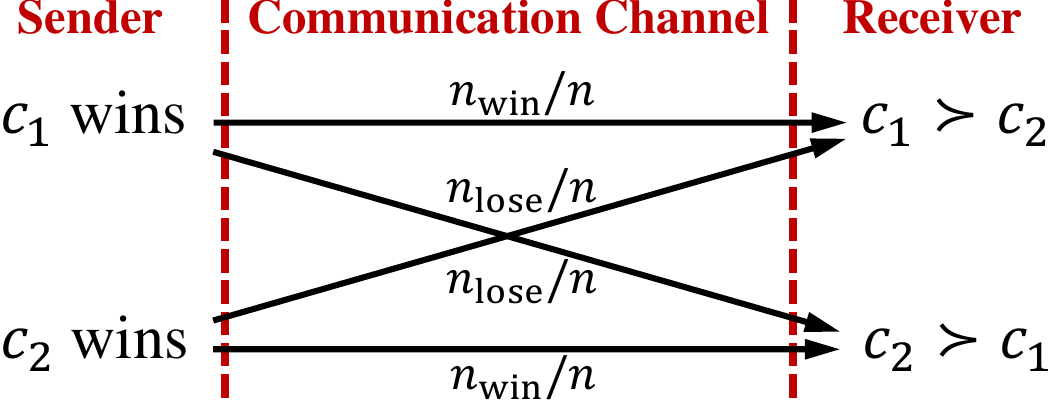}
    \caption{The communication channel presentation of sampling with replacement.}
    \label{fig:channel}
\end{figure}

According to Equation (1.35) in \citet{mackay2003information}, the capacity of the above communication channel $\cp = 1-H(\nwi/n)$, where $H:(0,1)\to(0,1]$ denotes the binary entropy function. Mathematically,
\begin{equation}\nonumber
    H(p) \triangleq -p\log(p)-(1-p)\log(1-p).
\end{equation}
\begin{proposition}
[$H(p)$'s Bounds, Theorem 1.2 in \citep{topsoe2001bounds}]\label{prop:H}
    Given any $p\in(0,1)$,
    \begin{equation}\nonumber
       4p(1-p) < H(p) < \big(4p(1-p)\big)^{1/\ln4}.
    \end{equation}
\end{proposition}
With the lower bound in Proposition~\ref{prop:H}, we know the communication channel's capacity
\begin{equation}\nonumber
\begin{split}
\cp &= 1-H(\nwi/n) \leq 1-H(1/2+\MOV/n) \\
&< 1-4\cdot(1/2-\MOV/n)\cdot(1/2+\MOV/n)\\
&= 4\MOV^2/n^2.
\end{split}
\end{equation}
The ``$\leq$'' follows by the monotonicity of $H(p)$. 
The next proposition (the well-known Shannon's theorem) connects the channel capacity with the error probability of binary information.
\begin{proposition}[\citep{shannon1948mathematical}]\label{prop:shannon}
Given a communication channel with capacity $\cp$, reconstructing each single-bit message with error probability $\varepsilon\in(0,0.5]$ requires receiving at least $\frac{1-H(\varepsilon)}{\cp}$ bits (in expectation) from the channel.
\end{proposition}
By Proposition~\ref{prop:shannon} and the upper bound in Proposition~\ref{prop:H}, the required number of bits from the channel (the required number of samples)
\begin{equation}\nonumber
\begin{split}
T &\geq  \frac{1-H(\varepsilon)}{\cp} > \frac{n^2\cdot\left(1-\big(4\varepsilon(1-\varepsilon)\big)^{1/\ln4}\right)}{4\MOV^2}\\
&= \Omega\left(\frac{n^2\cdot\left(\frac{1}{2}-\varepsilon\right)^2}{\MOV^2}\right).
\end{split}
\end{equation}
Lemma~\ref{lem:classical_m2} follows by the observation that the time-complexity and the space-complexity of getting $T$ samples are $\Omega(T)$ and $\Omega(\log T)$ respectively.
\end{proof}

\subsection{Compare sampling with and without replacement}\label{app:with_witout}
Although Theorem~\ref{thm:classical} holds only for sampling with replacement algorithms, we believe that when the algorithm only uses the histogram of the sample votes to calculate the winner, and the sampled size $T$ is small compared to $n$, then there is no major difference for sampling without replacement algorithms, because two samplings will converge to the same distribution when $n$ goes to infinity. 
 
Let $\hist$ be the histogram for a profile $P$, and $\hist_j$ is the number of votes for $j$-th ranking in the profile. In the sampling with replacement, the number of votes for $j$-th ranking in the sample follows binomial distribution $\mathcal{B}(T, \hist_j / n)$.
For the sample without replacement, the number of votes for $j$-th ranking follows hypergeometric distribution $\mathcal{H}(n, \hist_j, T)$. (A hypergeometric distribution $\mathcal{H}(n, \hist_j, T)$ considers drawing $T$ samples from $n$ items, among which exactly $\hist_j$ items have a specific feature, and characterizes the probability that a certain number of featured items is sampled.) The following proposition tells us that hypergeometric distribution $H(n, \hist_j, T)$ converges to binomial distribution $\mathcal{B}(T, \hist_j / n)$ when $n\to \infty$. 


\begin{theorem}[Corollary 4.1 in~\citep{Teerapabolarn2011:pointwise}.]\label{thm:with_without}
    Let $X$ be a random variable that follows hypergeometric distribution $\mathcal{H}(n, \hist_j, T)$, and $Y$ be a random variable that follows binomial distribution $\mathcal{B}(T, \frac{\hist_j}{n})$. For any $t\in\{0,\cdots, T\}$, fixed $p = \frac{\hist_j}{n}$, and $T = o(\sqrt{n})$, $\lim_{n\to\infty} |P(X = t) - P(Y = t)| = 0$. 
\end{theorem}

Therefore, when $n$ is large and sampling size $T$ is small compared to $n$, the sample histograms will be close to each other between sampling with and without replacement.

\section{Additional Experiments}
\subsection{Implementation Details}\label{app:detail_exp}
For the classical algorithm, we use MATLAB's built-in function  \texttt{mnrnd} to draw samples for $\hhist$ (follows multi-nominal distribution). For the quantum algorithm, we first calculate the distribution of quantum counting according to (5.26) in \citet{nielsen2002quantum} and then draw samples from the calculated distribution. For all experiments of this paper, we use $10^5$ independent trails to estimate $\prc$. All experiments of this paper are implemented through MATLAB 2022b and run on a Windows 11 desktop with AMD Ryzen 9 5900X CPU and 32GB RAM. 

\subsection{Additional experimental results}\label{app:exp}
\textbf{Plurality. } For plurality, we use the following profile $P$, \begin{equation}\nonumber\left\{
\begin{array}{l}
\frac{n+2(m!-1)\MOV}{m!} \text{ votes for } c_1\succ \cdots \succ c_m\\
\\
\frac{n-2\MOV}{m!} \text{ votes for each other type of votes}
\end{array}
\right..
\end{equation}
It's easy to check that the margin of victory of the above profile is $\MOV$ under plurality. Figure~\ref{fig:bor_m_4} plots the comparison between quantum-accelerated voting and classical voting for $m=4$. Similar acceleration as $m=2$ can be observed for $m=4$.

We also observe that that $\Pr[\text{correct}]$ may not monotonically increase with the increase of $\log_2(K\cdot 2^{\tone})$. \emph{e.g.,} for Figure~\ref{fig:plu_m_4}, $K=1$, and $\MOV=256$, the $\Pr[\text{correct}]$ for $s=15$ is smaller than $s=14$. The non-monotonicity is not an uncommon phenomenon in quantum algorithms (\emph{e.g.,} \citealp{kerenidis2019q,chen2020low,bausch2020recurrent}). This phenomenon comes from the discrete manner of quantum noises, which differs from the noise in classical sampling. We also note that our theoretical analysis bounds the asymptotic manner of $\Pr[\text{correct}]$, instead of the monotonicity. To be slightly more technical, this decrease comes from the noise (\emph{i.e.} tail probability) of the quantum counting, which is different from the classical counting noise. The tail probability of quantum counting also depends on the relative distance between the ground truth and its best $s$-bit estimation. The closer it is, the smaller the tail probability is. The relative distance may not monotonically decrease with the increase of $s$. For example, assume the ground truth of $\phi$ is $0.0001$ (in binary decimal). If using 1-bit estimation, the relative distance is $(0.0001-0.0)/0.1 = 1/8$. However, if using 3-bit estimation, the relative distance becomes $(0.0001-0.0)/0.001 = 1/2$, which is much larger than $1/8$.

\textbf{Borda.} For Borda, we let $d = \frac{4\MOV}{(m-2)!\cdot m}$ and set the profile as
\begin{equation}\nonumber
\left\{
\begin{array}{ll}
\frac{n+(m-1)d}{m!} \begin{array}{ll}\text{ votes for each type such that }\\
 ~c_1 \text{ is top-ranked}
 \end{array}\\
\\
\frac{n-d}{m!}  \text{ votes for each other type of votes}
\end{array}
\right..
\end{equation}
It's easy to check that the margin of victory of the above profile is $\MOV$ under Borda. Figure~\ref{fig:bor_m_4} plots the comparison between quantum-accelerated voting and classical voting for $m=4$. Similar behavior as plurality can be observed for Borda. 
\begin{figure*}[htp]
    \centering
    \includegraphics[width = 0.99\textwidth]{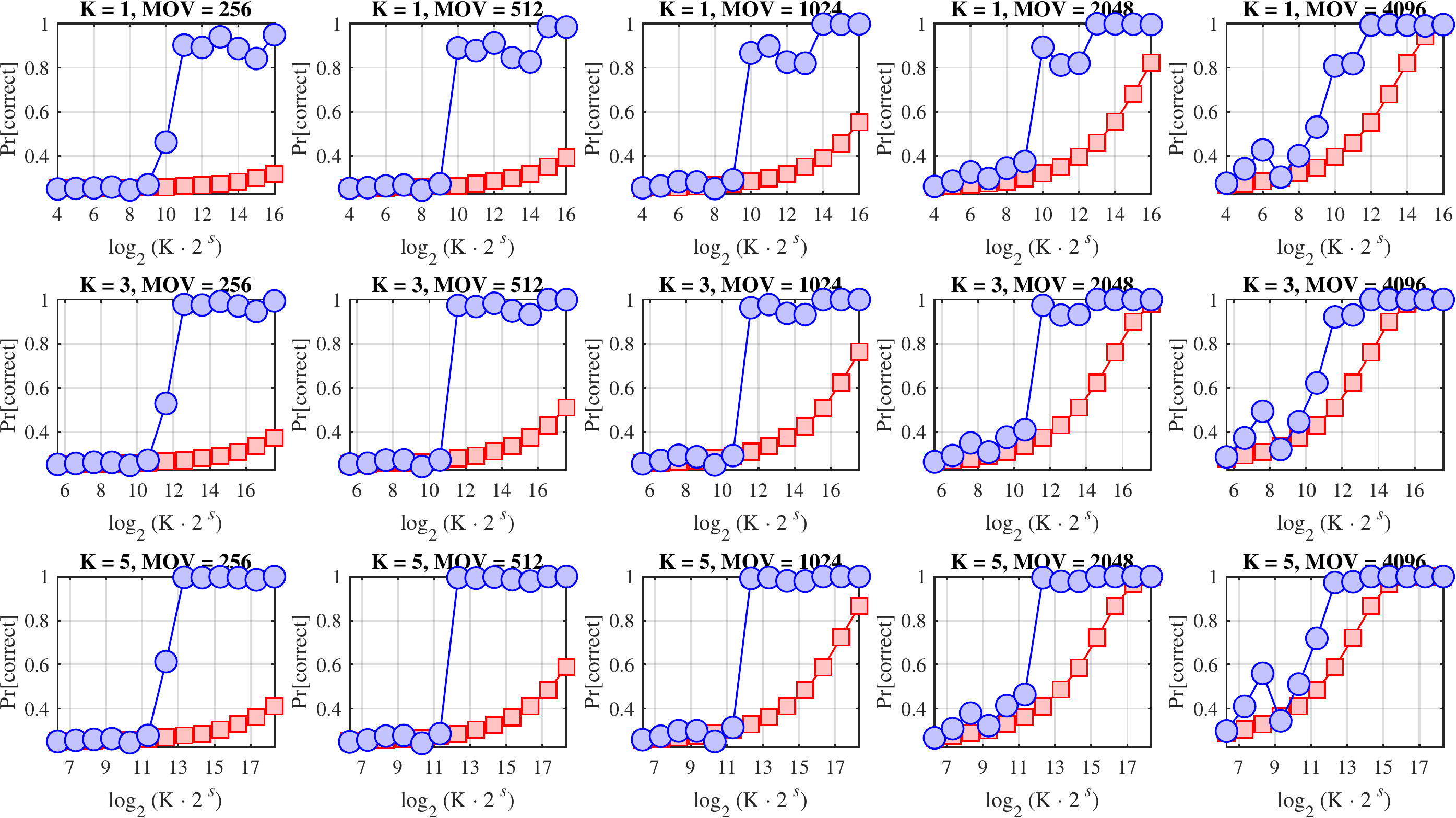}
    \caption{Compare quantum-accelerated voting (blue circles) with classical fast voting (red squares) for plurality when $m=4$. The horizontal axis can be seen as the logarithm of the algorithms' runtime.}
    \label{fig:plu_m_4}
\end{figure*}

\begin{figure*}[htp]
    \centering
    \includegraphics[width = 0.99\textwidth]{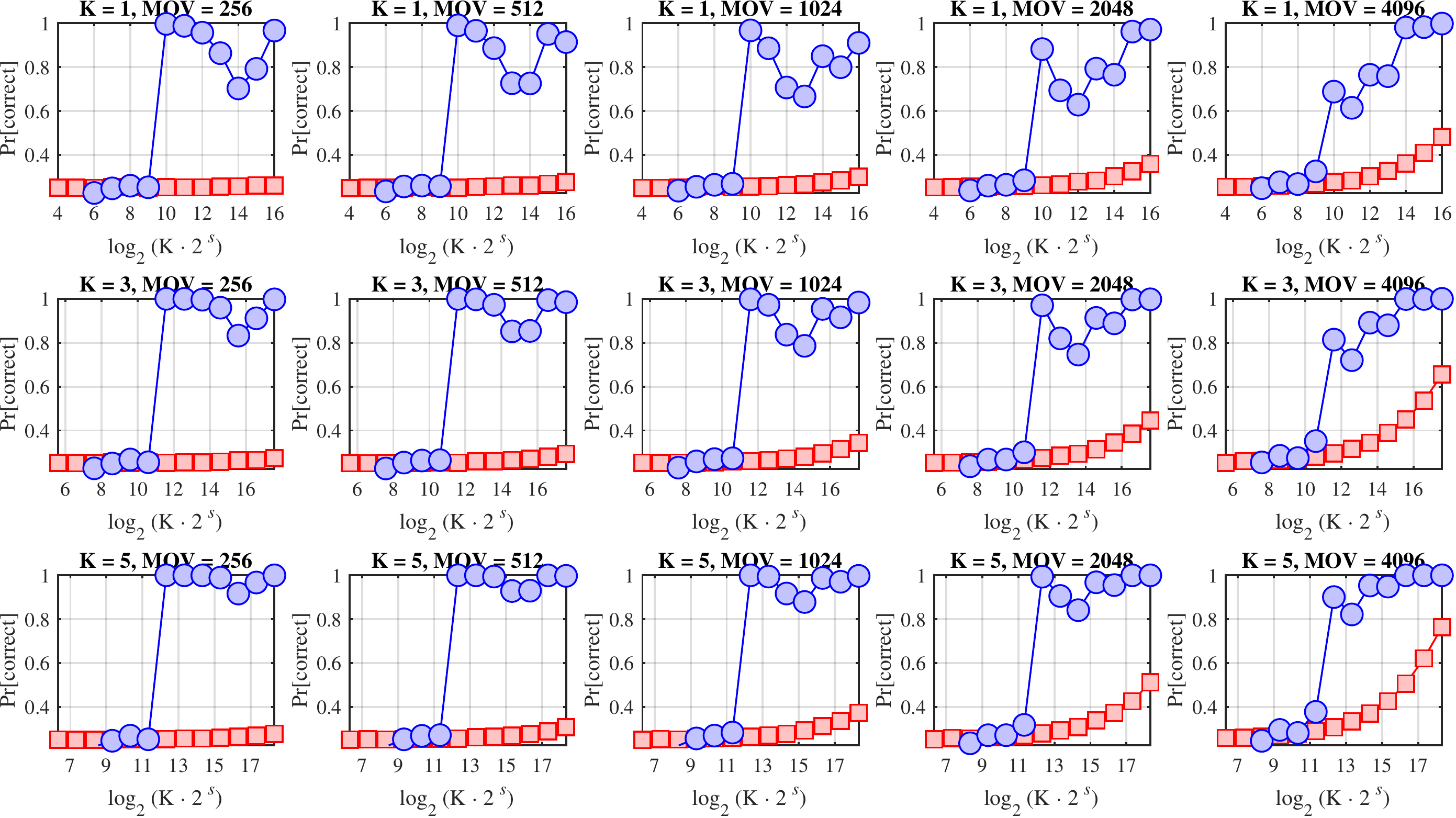}
    \caption{Compare quantum-accelerated voting (blue circles) with classical fast voting (red squares) for Borda when $m=4$. The horizontal axis can be seen as the logarithm of the algorithms' runtime.}
    \label{fig:bor_m_4}
\end{figure*}

\textbf{Copeland. }
For Copeland, we set the profile as
\begin{equation}\label{equ:cope}
\left\{
\begin{array}{ll}
\frac{n-2\MOV}{m!}+\frac{2\MOV}{(m-2)!} \begin{array}{ll}\text{votes for each type in the }\\
\text{form of } c_1\succ c_2\succ \text{others}
 \end{array}\\
\\
\frac{n-2\MOV}{m!}  \text{ votes for each other type of votes}
\end{array}
\right..
\end{equation}
It's easy to check that the margin of victory of the above profile is $\MOV$ under Copeland. Figure~\ref{fig:copeland} plot the comparison between quantum-accelerated voting and classical fast voting for $m=4$. Similar behavior as plurality can be observed for Copeland. 

\begin{figure*}[htp]
    \centering
    \includegraphics[width = 0.99\textwidth]{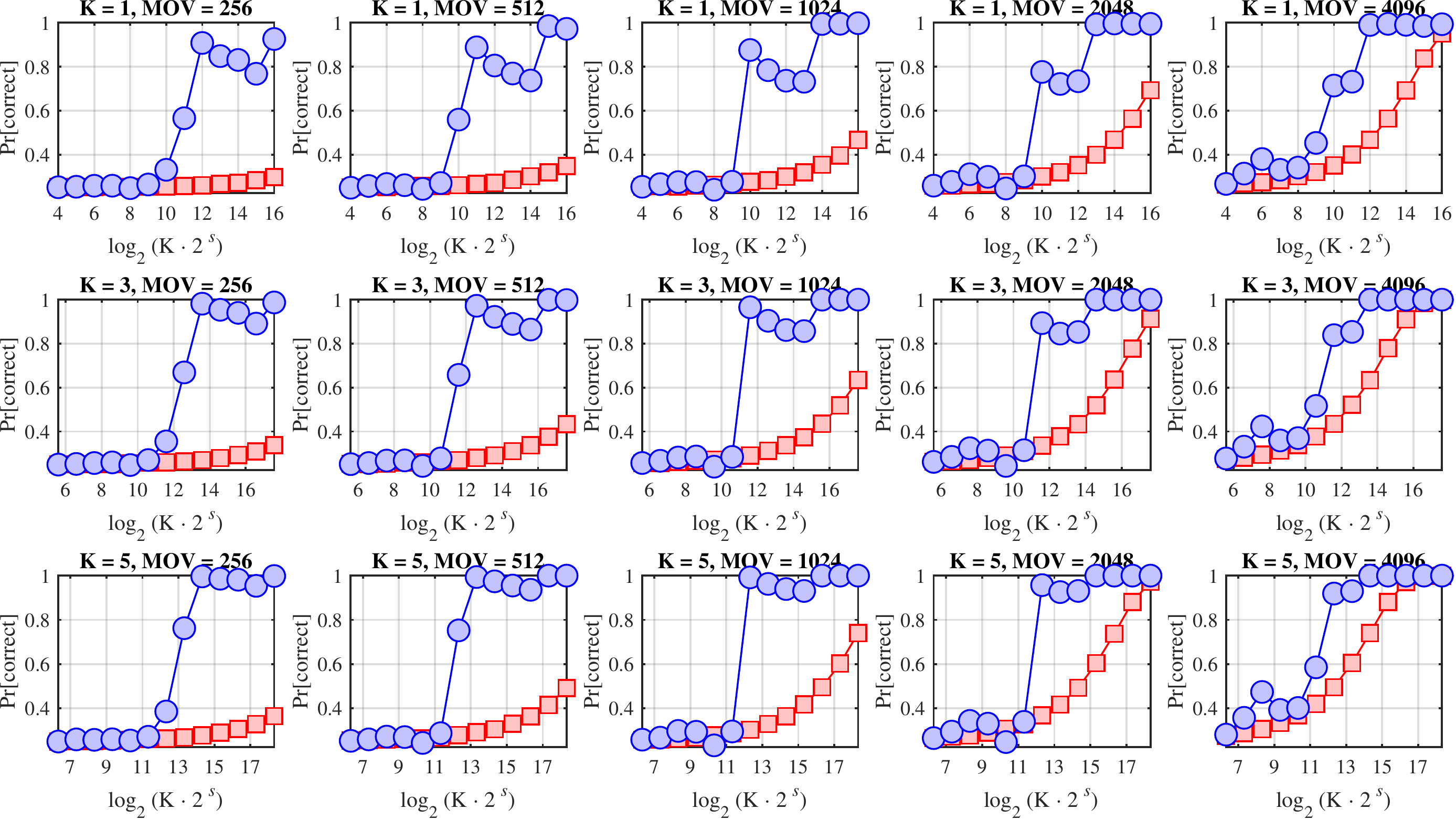}
    \caption{Compare quantum-accelerated voting (blue circles) with classical fast voting (red squares) for Copeland when $m=4$. The horizontal axis can be seen as the logarithm of the algorithms' runtime.}
    \label{fig:copeland}
\end{figure*}

\textbf{Single transferable vote (STV).}
For STV, the same profile as Copeland (see Equation (\ref{equ:cope})) is used. It's easy to check that the margin of victory of the profile is $\MOV$ under STV. Figure~\ref{fig:stv} plot the comparison between quantum-accelerated voting and classical voting for $m=4$. Similar behavior as plurality can be observed for STV. 

\begin{figure*}[htp]
    \centering
    \includegraphics[width = 0.99\textwidth]{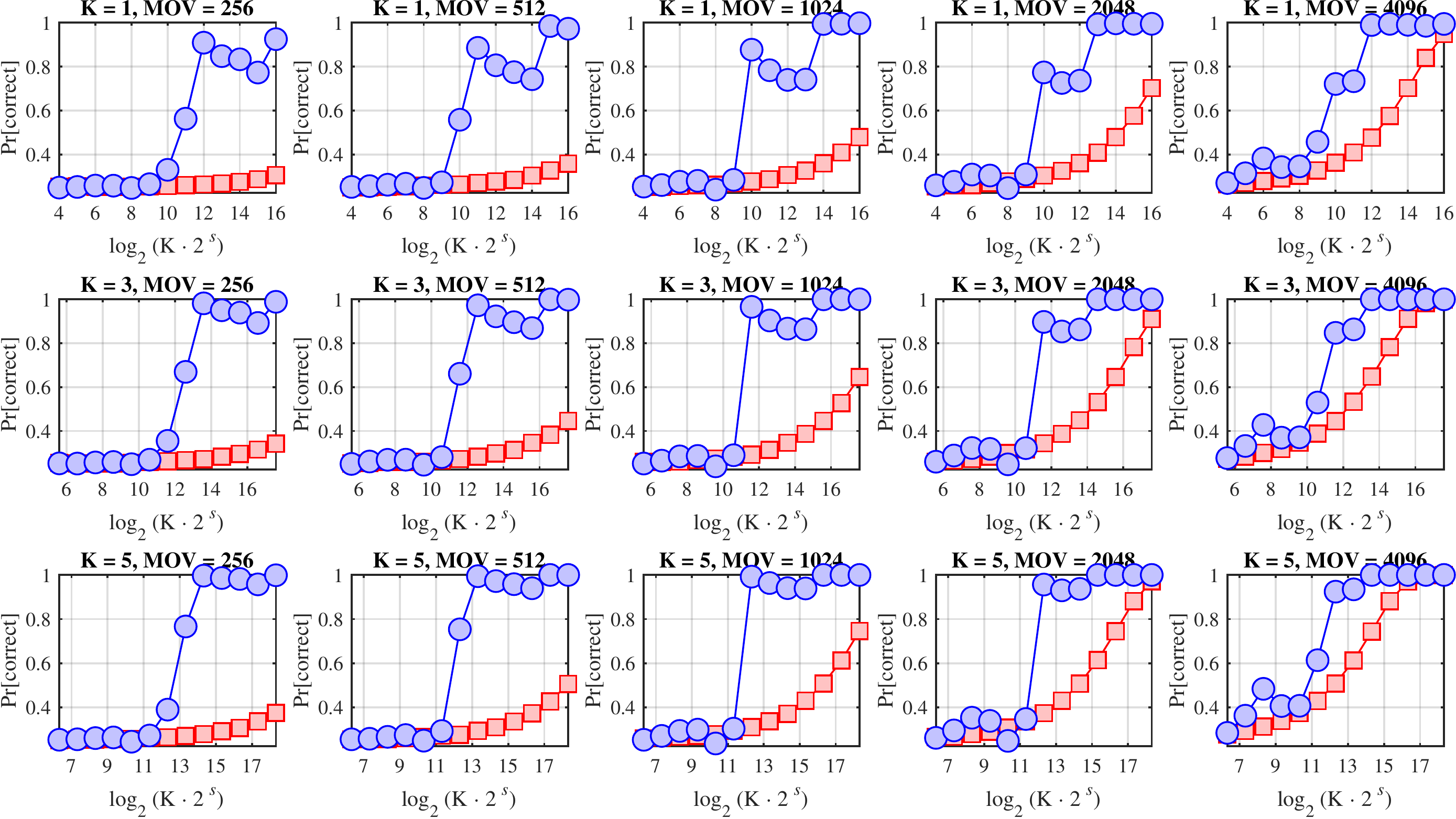}
    \caption{Compare quantum-accelerated voting (blue circles) with classical fast voting (red squares) for STV when $m=4$. The horizontal axis can be seen as the logarithm of the algorithms' runtime.}
    \label{fig:stv}
\end{figure*}

\textbf{Additional notes. } Since Copeland and STV shares the same profile, Figure~\ref{fig:copeland} and Figure~\ref{fig:stv} look similar. However, they are not the same and some small differences can be observed between the two figures. We also note that all four voting rules (plurality, Borda, Copeland, and STV) reduce to the majority voting when $m=2$. 


\end{document}